\documentclass[11pt]{article}
\usepackage{amsmath,amssymb,amsthm,natbib,graphicx,booktabs,url,multirow}
\usepackage[margin = 3cm]{geometry}
\theoremstyle{plain}  
\newtheorem{thm}{Theorem}[section] 
\newtheorem{prop}[thm]{Proposition}

\newtheorem{definition}[thm]{Definition}
\newtheorem{rem}[thm]{Remark}
\usepackage{setspace}

\title{A Kernel Score Perspective on Forecast Disagreement\\ and the Linear Pool}
\author{Fabian Krüger\thanks{Karlsruhe Institute of Technology, Institute of Statistics, fabian.krueger@kit.edu. I thank Sam Allen and seminar participants at Berlin (Ökonometrischer Ausschuss) and Erasmus University Rotterdam for helpful comments.}}

\begin{document}

\maketitle

\begin{abstract}
This paper generalizes several results on linear pooling from squared error loss to all kernel scores. The latter are a rich family of scoring rules that covers point and distribution forecasts for univariate and multivariate, discrete and continuous settings. Its members include the Continuous Ranked Probability Score for univariate distribution forecasting and the Energy Score for multivariate distribution forecasting. 
Our results indicate that forecast disagreement (measured as the average pairwise divergence of all component distributions) has important implications for the linear pool's performance. The results are useful for understanding and designing linear pools in general combination settings. In particular, they motivate using the linear pool (as opposed to other combination formulas) and yield a novel condition under which equal combination weights are optimal under a given kernel scoring rule. 
\end{abstract}

\section{Introduction}
\label{sec:intro}
Forecast combination is a widely popular method that applies to various types of forecasts, including point forecasts and forecast distributions. Empirically, combinations have repeatedly been found to perform well relative to using any individual forecasting method \citep{WangEtAl2023}. 

Combinations of forecast distributions often take an appealingly simple linear form proposed by \cite{Stone1961}. This form is also known as the linear (prediction) pool. A strand of literature including \cite{genest1986combining}, \cite{HallMitchell2007}, \cite{GewekeAmisano2011}, \cite{GneitingRanjan2013}, \cite{LichtendahlEtAl2013}, \cite{BerrischZiel2023} and \cite{Taylor2025} analyzes the linear pool, and compares it to other combination methods such as quantile-based approaches.

The present paper contributes to this literature by studying the linear pool for kernel scores \citep{Dawid2007,GneitingRaftery2007}, a rich class of proper scoring rules that covers point or distribution forecast of discrete or continuous, univariate or multivariate outcomes. Table \ref{tab:scores} lists some prominent examples, including squared error for point forecasts and the Continuous Ranked Probability Score \citep[CRPS;][]{MathesonWinkler1976} for distributions. Our analysis shows that several well-known results on linear pooling under squared error generalize to all kernel scores. This is potentially surprising since kernel scores cover settings that seem distinct from an applied perspective.

\begin{table}
\begin{tabular}{lll}
Scoring rule & Outcome type & Interested in \\ \toprule
Squared error (SE) & Univariate, quantitative & Mean \\
Multivariate SE  & Multivariate, quantitative & Mean \\
CRPS & Univariate, quantitative & Full distribution \\
Energy Score  & Multivariate, quantitative & Full distribution \\	
Brier Score & Univariate, categorical (unordered) & Full distribution  \\
Ranked Probability Score (RPS) & Univariate, categorical (ordered) & Full distribution \\ \bottomrule
\end{tabular}
\caption{Examples of kernel scores. See Appendix \ref{sec:specific} for details. \label{tab:scores}}
\end{table} 

Our results can be grouped into two categories. The first group describes mechanical properties of forecast combination under any kernel score. In particular, Proposition \ref{prop:d} shows that the linear pool maximizes centrality among the $n$ component distributions when distance is measured in terms of kernel scores. This motivation distinguishes the linear pool from other forms of combination such as quantile averaging. We further argue that the linear pool's expected score (i.e., its entropy) can usefully be decomposed as the sum of two components capturing (i) the average entropy of the component distributions, and (ii) disagreement between the component distributions. This disagreement term coincides with the performance difference between the pool and the component distributions (see Proposition \ref{prop1}). For kernel scores, disagreement hence quantifies the gains from forecast combination. 

Our second group of results derives implications for the performance of linear pools. In Proposition \ref{prop:equidist}, we present sufficient conditions for optimality of equal weights, nesting existing conditions for point forecasts. Motivated by the specific structure of kernel scores, we consider dynamic weights that reflect information on forecast disagreement, instead of being constant over time (Remark \ref{remark:csr}). We further show that a linear pool of calibrated componenents is underconfident, with the amount of underconfidence depending on forecast disagreement (Proposition \ref{prop:under}). In order to put the properties of linear pooling in perspective, we present novel results on quantile-based forecast combination which has become a popular alternative to linear pooling \citep[see e.g.][]{FakoorEtAl2023}. While quantile-based combination can have better calibration properties than linear pooling (Proposition \ref{lpvsqp}), a quantile-based combination of auto-calibrated components fails to be auto-calibrated (Proposition \ref{prop:qpac}), mirroring existing results for linear pooling. 

In macroeconomics and finance, many studies consider notions of disagreement to either measure economic conditions (see \citealt{ClarkMertens2024} and the references therein), or to test economic theories (see e.g. \citealt{CoibionEtAl2012}). As a side effect, the results of this paper suggest principled measures of forecast disagreement for various applied settings. In particular, we cover the case of disagreement among probability distributions (rather than point forecasts) recently studied by \cite{CumingsEtAl2021} and \cite{MitchellEtAl2024}. In Section \ref{sec:illustration}, we illustrate a new measure of disagreement among bivariate forecast distributions that are based on either surveys or a Bayesian vector autoregressive model. 

The present paper complements recent independent work by \cite{AllenEtAl2024} on linear pooling for kernel scores. They consider computational aspects (see their Proposition 2), a linear pool variant based on order statistics (their Section 4.2) and empirical analysis of weather forecasts (their Section 5). By contrast, our focus is on generalizing forecast combination results obtained under squared error loss, and the themes we study -- in particular, forecast disagreement -- are motivated by economic applications.

The paper is structured as follows: Section \ref{sec:setup} describes our formal setup. Section \ref{sec:d_motiv} presents a disagreement-based motivation of the linear pool. Section \ref{sec:entropy} presents a decomposition of the linear pool's entropy. Section \ref{sec:performance} relates the linear pool's ex-post performance to its disagreement component. Section \ref{sec:implications} studies statistical implications of the aforementioned results. Section \ref{sec:simulation} presents a Gaussian example. Section \ref{sec:empirical} provides empirical illustrations using bivariate forecast distributions for inflation and univariate forecast distributions for exchange rates. Section \ref{sec:discussion} concludes with a discussion. The appendix contains specific formulas for various scoring rules, proofs, computational details and additional figures. \textsf{R} code accompanying the paper is available at \url{https://gitlab.kit.edu/fabian.krueger/kernel_pool_replication}.

\section{Formal Setup} \label{sec:setup}

\subsection{Scoring Rules}

A scoring rule $S: \mathcal{F} \times \Omega \rightarrow \mathbb{R} \cup \{\infty\}$ assigns a numerical score, given a forecast distribution $F \in \mathcal{F}$ and an outcome $y \in \Omega$. We use scoring rules in negative orientation, such that a smaller score indicates a better forecast. Throughout this paper, we assume that the class $\mathcal{F}$ of distributions is convex, such that it contains any linear combination $w F_1 + (1-w)F_2$ of two members $F_1, F_2$, where $0 \le w \le 1$. 

Suppose that the predictand is distributed according to $F \in \mathcal{F}$. Then the expected score when stating the forecast $H$ is given by 
$$\mathbb{E}_F[S(H, X)] = \int_\Omega S(H, x)~\text{d}F(x),$$
where $\mathbb{E}_F$ denotes expectation with respect to a random variable (in this case, $X$) that is distributed according to $F$. Throughout much of this paper, we will treat distributions such as $F$ and $H$ as well as combination weights $w_1, w_2, \ldots, w_n$ as fixed, non-random objects. Exceptions occur in Sections \ref{sec:implications} and \ref{sec:simulation}, where we consider a `prediction space' framework in which forecasts and outcomes follow a joint distribution.  

The score divergence $d_S(H,F)$ measures the difference in expected scores when stating $H$, as opposed to the actual distribution $F$. It is given by 
\begin{eqnarray*}
d_S(H,F) &=& \mathbb{E}_F[S(H, X)]-\mathbb{E}_F[S(F, X)] \\
&=& \int_\Omega \left[S(H, x)-S(F,x)\right]~\text{d}F(x).
\end{eqnarray*}
A proper scoring rule $S$ satisfies $d_S(H, F) \ge 0$ for all $F, H \in \mathcal{F}$. Proper scoring rules thus incentivize truthful and accurate forecasting: A forecaster whose beliefs are represented by $F$ cannot do (strictly) better by reporting a forecast distribution other than $F$. \cite{ThorarinsdottirEtAl2013} argue that divergence functions $d_S(H,F)$ based on proper scoring rules are useful for measuring the discrepancy between two distributions $F$ and $H$ in applications. Below we will argue that divergence functions based on kernel scores, a subfamily of proper scores, are particularly useful due to their symmetry property that $d_S(H,F) = d_S(F,H)$. 

The entropy $\mathbb{E}_F[S(F,X)]$ represents a forecaster's expected score when reporting $F$. Given that $F$ is an optimal forecast, this is the best (i.e., smallest) expected score that can be attained. Entropy hence measures the uncertainty implicit in $F$. In the following, we will occasionally use the shorter notation $\text{ent}(F) := \mathbb{E}_F[S(F,X)]$ to denote the entropy of $F$.

\subsection{Kernel Scores}\label{sec:kernel}

We next describe kernel scores, a rich family of scoring rules proposed by \cite{GneitingRaftery2007} and \cite{Dawid2007} that relates to statistical concepts of energy \citep{SzekelyRizzo2017} and to kernel methods in machine learning \citep[c.f.][Section 2.2]{AllenEtAl2024}. Appealingly, kernel scores can accommodate very general (Hausdorff) outcome spaces $\Omega$, containing e.g. univariate or multivariate, discrete or continuous outcome variables. 

Our setup mostly follows \cite{Gneiting2012}. Let $L: \Omega \times \Omega \rightarrow [0, \infty)$ be a nonnegative function that is symmetric in its two arguments, with $L(z, z) = 0$ for all $z \in \Omega$, and the property that 
$$\sum_{i=1}^n\sum_{j=1}^n c_ic_j L(x_i, x_j) \le 0$$
for all $n \in \mathbb{N}, x_1, x_2, \ldots, x_n \in \Omega$ and $c_1, c_2, \ldots, c_n \in \mathbb{R}$ such that $\sum_{i=1}^n c_i = 0$. The function $L$ is called a negative definite kernel. Based on $L$, one can construct the scoring rule 
\begin{equation}
S_L(F, y) = \mathbb{E}_F[L(X, y)]- \frac{1}{2} \mathbb{E}_F[L(X, \widetilde{X})],\label{score}
\end{equation}
where $X$ and $\widetilde{X}$ are understood to be two independent draws from $F$. Scoring rules from the family in (\ref{score}) are proper with respect to the class $\mathcal{F}$ of Radon probability measures on $\Omega$ for which the expectation $\mathbb{E}_F[L(X, \widetilde X)]$ is finite \citep[p.~15]{Gneiting2012}. Under these conditions, kernel scores satisfy $S_L(F, y) \ge 0$ for all $F \in \mathcal{F}$ and $y \in \Omega$, with $S_L(F,y) = 0$ if $F = F_\delta$ has point mass at $y$ \citep[Theorem 2.4]{Gneiting2012}.

As detailed in Appendix \ref{sec:specific}, the family at (\ref{score}) includes popular scoring rules such as the squared error, Brier score, Continuous Ranked Probability Score (CRPS) or Energy Score, which correpond to specific choices of $\Omega$ and $L$. For a scoring rule $S_L$ as in (\ref{score}), the entropy function is given by 
\begin{equation}
\mathbb{E}_F [S_L(F, X)] = \frac{1}{2} \mathbb{E}_F[L(X, \widetilde{X})].\label{entropy}
\end{equation}
The special case where $F = F_\delta$ places point mass on a single element $y \in \Omega$ results in the minimal entropy of zero.

For a kernel scoring rule, the divergence between two distributions $H, F$ is given by
\begin{equation}
d_{S_L}(H,F) = E_{F, H}[L(X, \widetilde{X})] - \frac{1}{2} \mathbb{E}_H[L(X,\widetilde{X})] - \frac{1}{2} \mathbb{E}_F[L(X,\widetilde{X})]; \label{divergence}
\end{equation}
here $E_{F, H}[L(X, \widetilde{X})]$ indicates that the expected value is with respect to two independent random variables $X \sim F$ and $\widetilde{X} \sim H$. Importantly, this divergence function is symmetric with respect to the two distributions $F$ and $H$. 

While rich, the family of kernel scores at (\ref{score}) does not include all proper scoring rules that are popular in practice. 
In particular, the tick (or pinball, or piecewise linear) loss function underlying quantile regression \citep{KoenkerBassett1978} features an asymmetric divergence function and is thus not a kernel score. Interestingly, the divergence function of the quantile score is asymmetric even in the case of the median, so that absolute error is not a kernel score.\footnote{To see this, consider an example where $F$ is a standard normal distribution and $H$ is a normal distribution with mean $1$ and variance $2$. Then $d_S(H,F) = \int_{\Omega} \left(|x-1|-|x|\right) \text{d}F(x) \neq d_S(F,H) = \int_{\Omega} \left(|x|-|x-1|\right) \text{d}H(x),$ i.e. the divergence function is asymmetric. \cite{BentzienFriederichs2014} provide an expression of the divergence function for general quantile levels, which is also asymmetric in general.} Furthermore, the logarithmic score is not a kernel score \citep{AllenEtAl2023}. This can again be deduced from the fact that its (Kullback-Leibler) divergence function is asymmetric.

Our definition of kernel scores at (\ref{score}) in terms of a negative definite kernel $L$ differs from \citeauthor{AllenEtAl2024}'s definition in terms of a positive definite kernel (see their Equation 4). Both definitions can yield equivalent expressions of a given kernel score like the CRPS, but the respective kernel functions differ. Using their definition, empirical weight optimization problems directly takes the form of a convex quadratic program (see their Proposition 2), which is convenient for software implementation. At the same time, our definition yields more easily interpretable connections to forecast disagreement, a key theme of this paper. In Appendix \ref{emp_weight}, we show how to translate our weight optimization problem into a convex quadratic format.

\section{A Divergence-Based Motivation for Linear Pooling}\label{sec:d_motiv}

This section uses the notion of divergence to motivate linear pooling, as opposed to other forms of forecast combination. Here we initially consider a finite outcome space $\Omega$, with $|\Omega| = n_\Omega$. The elements of $\Omega$ could be univariate or multivariate, quantitative or categorical. As noted by \citet[Section 4]{AllenEtAl2024}, a finite outcome space considerably simplifies the expected value expressions relevant for kernel scores, and aligns well with the fact that many forecasting models (such as meteorological ensembles or models estimated via Bayesian techniques) take the form of simulated or empirically observed samples. Therefore, and since $n_\Omega$ can be arbitrarily large, the assumption of a finite outcome space is not very restrictive from an applied perspective. 

We identify a forecast distribution $F$ with an $n_\Omega \times 1$ vector $p$ containing predicted probabilities of all outcomes. Accordingly, the family $\mathcal{F}$ of forecast distributions consists of the $n_\Omega$-dimensional probability simplex $\text{PS}^{n_\Omega}$, i.e. the set of all $n_\Omega$ dimensional vectors whose elements are nonnegative and sum to one. Let $F_w = \sum_{i=1}^n w_i F_i$ be a linear pool of $n$ forecast distributions $(F_i)_{i=1}^n$, with $F_i \in \mathcal{F}$ and weight $w_i$ placed on the $i$th component, such that $w = (w_1, w_2, \ldots, w_n)' \in \text{PS}^n$. Clearly, $F_w \in \mathcal{F}$.

We further define the matrix $\underline{L}$ whose $[j,l]$ element is given by $L(x_{(j)},x_{(l)})$, where $x_{(j)}$ and $x_{(l)}$ are the $j$th and $l$th unique elements of $\Omega$. In this setup, we have
\begin{eqnarray}
	\mathbb{E}_{F,H}\big[L(X,\widetilde{X})\big] &=& p' \underline{L}~h, \label{example} 
\end{eqnarray}
where $h$ is the $n_\Omega$-vector of probabilities corresponding to some forecast distribution $H$. Based on (\ref{divergence}) and (\ref{example}), the average divergence between the component distributions $F_1, F_2, \ldots, F_n$ and $H$ is given by
\begin{eqnarray}
	D_\text{gen}(h) &= &\sum_{i=1}^n w_i \left\{{p_i}' \underline{L} h - \frac{1}{2} h' \underline{L} h - \frac{1}{2} {p_i}' \underline{L} p_i\right\};
\end{eqnarray}
where the vectors $p_1, \ldots, p_n$ correspond to the component distributions, and the notation $D_\text{gen}(h)$ indicates disagreement around a generic vector $h$ of probabilities. 

We next ask which vector $h$ of probabilities minimizes $D_\text{gen}(h)$. In the context of proper (but not necessarily kernel) scoring rules, \cite{NeymanRoughgarden2023} call this minimizer the `quasi-arithmetic pool'. The latter is the most central choice of probabilities, a desirable feature if the combination aims to represent a consensus of the individual forecasts $F_1, F_2, \ldots, F_n$. If $S_L$ is the Brier score and the predictand is univariate and categorical, the quasi-arithmetic pool is known to coincide with the linear pool (\citealt[Proposition 3]{Pettigrew2019}; \citealt[Section 1.3.3]{NeymanRoughgarden2023}). The following result shows that the equivalence between the quasi-arithmetic pool and the linear pool generalizes to all kernel scores. 
\begin{prop}
	Assume that the scoring rule $S = S_L$ is a kernel score, and that the outcome space $\Omega$ is finite, with $|\Omega| = n_\Omega$. Then the linear pool $F_w$ minimizes the average divergence to its components. That is, $D_\text{gen}(h) -D_{\text{gen}}(p_w) \ge 0$ for every vector $h$ of probabilities over $\Omega$.\label{prop:d}
\end{prop}

In the special case where $S_L$ is squared error loss and the predictand is univariate and quantitative, Proposition \ref{prop:d} recovers the textbook result that the arithmetic mean minimizes the sum of squared errors.\footnote{Specifically, suppose that the $i$th component distribution has point mass at $x_i \in \Omega$, such that its cumulative distribution function is given by $F_i(z) = \mathbf{1}(z \ge x_i).$ Then the score divergence between $F_w$ and $F_i$ is given by $(\mu_w - x_i)^2$, with $\mu_w = \sum_{i=1}^n w_i x_i$ being the mean implied by $F_w$.} Interestingly, the proposition also covers multivariate quantitative outcomes for which several kernel scores $S_L$ are available (most popularly, the Energy Score). The equivalence between the quasi-arithmetic and linear pools need not hold for proper scoring rules $S$ that are not kernel scores. In particular, if $S$ is the logarithmic scoring rule, the quasi-arithmetic pool is given by the logarithmic pool \citep[Section 1.3.4]{NeymanRoughgarden2023}. However, the latter result hinges on the order of the two arguments of the relevant (Kullback-Leibler) divergence function. When reversing the order of the arguments, the linear pool minimizes the average Kullback-Leibler divergence to the component distributions (\citealt{Abbas2009}, Proposition 1; \citealt{Pettigrew2019}, Proposition 3). This type of sensitivity is a drawback of using an asymmetric divergence function. Conversely, the symmetry of a kernel score's divergence function is appealing.

\section{Entropy of the Linear Pool}\label{sec:entropy}

We now consider linear pooling for a generic (not necessarily discrete) outcome space $\Omega$. As before, a linear pool is given by $F_w = \sum_{i=1}^n w_i F_i$, with $F_i \in \mathcal{F}$ and a weight vector $w \in \text{PS}^n$.

\subsection{Decomposition}

Let $S$ be a scoring rule, and let $\mathcal{F}$ be such that $S$ is proper with respect to $\mathcal{F}$. Then the entropy of $F_w$ can be written as
\begin{eqnarray}
\text{ent}(F_w)  &=& \sum_{i=1}^n w_i~ \mathbb{E}_{F_i}[S(F_w, X)]\nonumber\\ 
	&=& \sum_{i=1}^n w_i \left\{ \mathbb{E}_{F_i}[S(F_w, X)]- \mathbb{E}_{F_i}[S(F_i, X)]  \right\} + \sum_{i=1}^n w_i~ \mathbb{E}_{F_i}[S(F_i, X)] \nonumber\\
	&=& \underbrace{\sum_{i=1}^n w_i~d(F_w, F_i)}_{D = \text{average divergence}} ~+ \underbrace{\sum_{i=1}^n w_i~ \text{ent}(F_i)}_{\text{\text{average entropy of components}}},\label{entropy_pool} 
\end{eqnarray}	
where $D \ge 0$ since the weights $w_i$ are nonnegative and $S$ is proper.

We demonstrate in Appendix \ref{sec:specific} that in the case of squared error, Equation (\ref{entropy_pool}) recovers the famous decomposition of the linear pool's variance as discussed by \cite{wallis2005} and others. Furthermore, \citet[Section 3.4]{ShojaSoofi2017} derive the decomposition at (\ref{entropy_pool}) for the case of the logarithmic score. In this case, 
$S(F, y) = -\log f(y),$ where $f$ is the density associated with $F$, and the corresponding divergence function is known as Kullback-Leibler divergence. Apart from these two special cases for $S$, we are not aware that the decomposition at (\ref{entropy_pool}) has appeared in the literature. The decomposition is consistent with the inequality
\begin{equation}
	\underbrace{\mathbb{E}_{F_w} \left[S(F_w,X)\right]}_{=\text{ent}(F_w)} \ge~\sum_{i=1}^n w_i~\underbrace{\mathbb{E}_{F_i}[S(F_i, X)]}_{=\text{ent}(F_i)}.\label{concave_entropy}
\end{equation}
which holds for any proper scoring rule $S$ \citep[Section 2.1]{GneitingRaftery2007}.\footnote{ \cite{GneitingRaftery2007} define scoring rules in positive orientation, whereas we define them in negative orientation. Hence the word `convex' in their statement (on p.~362) that `a regular scoring rule $S$ is proper [..] if and only if the expected score function [..] is convex [..]' must be replaced by `concave' in our setting.} Equation (\ref{entropy_pool}) shows that the difference between the left and right sides of (\ref{concave_entropy}) is given by $D$. The following result describes some properties of the decomposition at (\ref{entropy_pool}) for the family of kernel scores. 
\begin{prop}
Let $S = S_L$ be a kernel scoring rule, and let $\mathcal{F}$ be such that $S_L$ is proper with respect to $\mathcal{F}$. Then
\begin{itemize}
\item[(a)] The divergences in the first term at (\ref{entropy_pool}) are symmetric, i.e. $d(F_w, F_i) = d(F_i, F_w)$. Furthermore, the pool's entropy $\text{ent}(F_w)$ and the entropy $\text{ent}(F_i)$ of each component are nonnegative. 
\item[(b)] The average divergence $D$ satisfies
\begin{eqnarray}
D &=& \sum_{i=1}^n w_i~d(F_i, F_w) \nonumber\\ 
 &=& \frac{1}{2}\mathbb{E}_{F_w}[L(X,\widetilde{X})] - \frac{1}{2} \sum_{i=1}^n w_i~\mathbb{E}_{F_i}[L(X, \widetilde{X})] \label{divergence_lp} \label{disagreement_lp} \\
&=& \sum_{i=1}^{n-1}\sum_{j>i} w_iw_j~d(F_i, F_j). \label{disagreement_pairs_lp}
\end{eqnarray}	
\end{itemize}
\label{prop0}
\end{prop}
The properties of kernel scores as noted in part (a) are appealing for interpreting the pool's entropy. In particular, the symmetry property is useful since the order of the arguments $F_w$ and $F_i$ to the divergence function $d$ seems arbitrary. \citet[Theorem 19]{WaghmareZiegel2025} show that under an additional technical assumption, kernel scores are the only proper scoring rules with a symmetric divergence function. Nonnegativity of entropy is intuitively appealing if one aims to interpret entropy as uncertainty, all common measures of which are nonnegative.\footnote{In principle, one could enforce positivity of any scoring rule by adding a large constant $C \in \mathbb{R}_+$, which would not affect propriety of the scoring rule. However, this transformation would be at odds with empirical practice and possibly cumbersome to implement.} The logarithmic score, which is not a kernel score as noted above, does not share the advantages mentioned in (a): Its (Kullback-Leibler) divergence function is asymmetric, and its entropy function can be negative. Part (b) of the proposition shows that the average divergence term $D$ can be written in terms of pairwise divergences between component distributions. This result is potentially helpful for computation, e.g. when the components are of a simpler functional form than their linear pool, as in the Gaussian case.

For kernel scores $S = S_L$, we thus argue that the term $D$ at (\ref{disagreement_lp}) and (\ref{disagreement_pairs_lp}) defines a useful measure of average disagreement within the linear pool. 

\subsection{Comparison to Quantile Combination} \label{sec:qp}

In the case of a univariate, continuous predictand, quantile combinations have become a popular alternative to the linear pool.  \citet[Lemma 1 and Proposition 8]{LichtendahlEtAl2013} show that forecast distributions obtained via quantile combination tend to be less dispersed than the linear pool in terms of convex order and (even) moments. We next relate their results to entropy as an alternative measure of uncertainty. To facilitate the comparison, we momentarily assume that the $n$ forecast distributions $F_1,\ldots,F_n$ are strictly increasing, so that the $u$-quantile of $F_w^q$ is unique and equals $\sum_{i=1}^n w_iF_i^{-1}(u),$ where $u \in (0,1)$. 

\begin{prop}
Consider the linear pool $F_w$ and the quantile combination $F_w^q$ using the same weight vector $w \in \text{PS}^n$. Furthermore, assume that the proper scoring rule $S$ is such that $S(F_w, y)$ is convex with respect to $y$ for all $y \in \Omega$. Then $\text{ent}(F_w) \ge \text{ent}(F_w^q)$, i.e. the linear pool's entropy exceeds the entropy of quantile combination.  \label{co_ent}
\end{prop}
If $S$ is squared error, the proposition's convexity assumption holds without any restriction on $F_w$, and the proposition recovers \citeauthor{LichtendahlEtAl2013}'s result that quantile combination yields a forecast distribution with smaller variance. For the CRPS, the convexity assumption also holds without any restriction on $F_w$. Note that the convexity assumption in Proposition \ref{co_ent} is sufficient but not necessary for higher entropy of the linear pool. Furthermore, while untypical in practice, it is possible to construct kernel score examples in which the linear pool's entropy is strictly smaller than that of quantile-based combination (see Appendix \ref{sec:proofs}).

\section{Disagreement and Forecasting Performance}\label{sec:performance}

Sections \ref{sec:entropy} considers the linear pool's entropy function, which captures the pool's assessment of its own forecasting performance (ex ante, that is, before the outcome has realized). The following result instead consider the pool's ex post performance (that is, after observing the outcome $Y = y$). 
\begin{prop}
For every kernel scoring rule $S_L$, the linear pool's score satisfies
$$S_L(F_w, y) = \underbrace{\sum_{i=1}^n w_i S_L(F_i, y)}_{\text{average score of components}}-~~D,$$
where $y \in \Omega$ is the realizing outcome, and $D  = \sum_{i=1}^n \sum_{j > i} w_iw_j~d(F_i, F_j) \ge 0$ defined at (\ref{disagreement_pairs_lp}) denotes the average divergence between the pool's components $(F_i)_{i=1}^n$ and the pool $F_w$. \label{prop1}
\end{prop}

For the special case where $S_L$ is squared error, the statement of the proposition is well known, dating back at least to \cite{Engle1983}. See \citet[Equation 9]{KnueppelKrueger2022} for details and discussion. For the special cases where $S_L$ corresponds to the CRPS, the statement of Proposition \ref{prop1} has been noted as Corollary 3.3.1 by \cite{Krueger2013}. 
Furthermore, Proposition \ref{prop1} sharpens Proposition 1 of \cite{AllenEtAl2024} which states that $S_L(F_w, y) \le \sum_{i=1}^n S_L(F_i, y).$ Finally, \cite{NeymanRoughgarden2023} consider the difference $\sum_{i=1}^n w_i S(F_i, y) - S(F_c, y)$ where $F_c$ is some (not necessarily linear) combination of the $n$ forecast distributions $F_1,F_2, \ldots, F_n$, and $S$ is a proper scoring rule. They establish a specific form of $F_c$ (`quasi-arithmetic pooling') that optimizes the difference in a worst-case sense (see their Theorem 4.1).\footnote{Since they define scoring rules in positive orientation, their expression for the difference in question must be multiplied by minus one in our context.} By contrast, our Proposition \ref{prop1} provides the specific form of the difference for the case that $F_c = F_w$ is the linear pool, and the scoring rule $S = S_L$ is a kernel score. 

Proposition \ref{prop1} implies that the linear pool improves upon the average performance of its components. The amount of improvement is given by disagreement, $D$. Hence, for a given average performance of the pool's components, and for given combination weights, it is desirable that the components be as diverse as possible. 

According to Proposition \ref{prop1}, the gains from linear pooling are constant, in that $S(F_w, y) - \sum_{i=1}^n w_i S(F_i, y) = D,$ where $D$ does not depend on $y$. Furthermore, our Proposition \ref{prop:d} states that when using a kernel score $S_L$, $F_w$ coincides with the quasi-arithmetic pool. Taken together, the two results hence imply that the gains from quasi-arithmetic pooling are constant in $y$. The latter statement is derived by \citet[Theorem 4.1]{NeymanRoughgarden2023} in a formal setup that is slightly different from ours,\footnote{They consider general proper scoring rules for categorical outcomes, whereas we consider kernel scores for general outcomes on a finite outcome space.} and using different proof techniques. \citeauthor{NeymanRoughgarden2023} also provide an appealing economic motivation for considering the gains from combination, in terms of the profit of an agent who subcontracts a group of expert forecasters. 

\section{Implications} \label{sec:implications}

In this section we discuss some implications of the previous results. For this purpose it will be necessary to take a stance on the properties of the forecast distributions $F_i$. We will describe these properties by means of a prediction space, a joint distribution of forecasts and realizations. \cite{GneitingRanjan2013} provide a general prediction space framework, formalizing ideas that date back to \cite{Bates-Granger-69} and \cite{MurphyWinkler1987}. Our following definition is a slightly adapted version of \citeauthor{GneitingRanjan2013}'s Definition 2.1. 

\begin{definition}
Let $n \ge 1$ be an integer. A prediction space is a probability space $(\Omega, \mathcal{A}, \mathbb{Q})$ together with sub-$\sigma$-algebras $\mathcal{A}_1,\mathcal{A}_2,\ldots,\mathcal{A}_n$, where the elements of the sample space $\Omega$ can be identified with tuples $(F_1,\ldots,F_n,Y)$ such that
\begin{itemize}
\item for $i = 1, \ldots, n$, $F_i$ is a CDF-valued random quantity that is measurable with respect to $\mathcal{A}_i$, and 
\item $Y$ is a real-valued random variable. 
\end{itemize}
\end{definition}

Let $\mathcal{L}(Y|\mathcal{A})$ denote the distribution of $Y$ given some sub-$\sigma$-algebra $\mathcal{A}$, and let $\sigma(F_i)$ denote the $\sigma$-algebra generated by the forecast $F_i$ itself. We then consider the following notion of forecast calibration  \citep{Tsyplakov2013, GneitingRanjan2013}:
\begin{definition}
The forecast distribution $F$ is auto-calibrated if $\mathcal{L}(Y|F_i) = F_i$ holds $\mathbb{Q}$-almost surely.
\end{definition}
Intuitively, the condition implies that $F_i$ is free of any systematic biases. A user who learns about the forecast $F_i$ (but has no access to other relevant information about $Y$) should thus use the forecast `as is', i.e. there is no need to correct or post-process the forecast in any way. See \citet[Section 2.1]{KnueppelEtAl2022} for further discussion.

\subsection{Choice of Combination Weights} \label{sec:implications_weights}

We first consider implications of Proposition \ref{prop1} on the choice of combination weights, a problem that has been widely discussed in the literature. In the context of mean forecasting under squared error loss, the empirical result that equal weights are often hard to beat (the `forecast combination puzzle') has motivated theoretical interest in the subject, see e.g. \cite{ClaeskensEtAl2016} and \cite{Elliott2025}. 

We initially assume that the weights $(w_i)_{i=1}^n$ are constant, i.e., they do not depend on conditioning information.
Proposition \ref{prop1} implies that the population version of the weight minimization problem is given by
\begin{equation}
w^* = \underset{w \in \text{PS}^n}{\text{arg min}} \left( \sum_{i=1}^n w_i \mathbb{E}_\mathbb{Q}(S_i) - \underbrace{\sum_{i=1}^{n-1} \sum_{j>i} w_iw_j~\mathbb{E}_\mathbb{Q}(d_{ij})}_{=\mathbb{E}_\mathbb{Q}(D)}\right),\label{constweights}
\end{equation}
where we use the simplified notation $S_i = S_L(F_i, Y)$ and $d_{ij} = d(F_i, F_j)$. The following result presents a sufficient condition under which equal weights are optimal, i.e. $w^* = (1/n, \ldots, 1/n)'$ at (\ref{constweights}). 
\begin{prop}\label{prop:equidist}
Suppose that a kernel scoring rule $S_L$ is used, and that the combination weights are required to be constant. Consider the following assumptions.
\begin{itemize}
\item[A] All $n$ individual forecasts perform equal on expectation, i.e. $\mathbb{E}_\mathbb{Q}(S_i) = c$ for all $i = 1,\ldots,n$. 
\item[B] The expected divergence between any pair of two forecast distributions $F_i, F_j$ is identical, i.e. $\mathbb{E}_\mathbb{Q}(d_{ij}) = d$ for all $i,j \in \{1,\ldots,n\}$ with $i \neq j$.
\end{itemize}
Under Assumption A, the population optimal combination weights maximize $\mathbb{E}_\mathbb{Q}(D)$ defined at (\ref{constweights}). Under Assumptions A and B, equal combination weights are optimal in population.
\end{prop}

If the forecasts are exchangeable in the sense that the joint distribution of $(F_1, \ldots, F_n, Y)$ is invariant to permutations of $F_1, \ldots, F_n$, then Assumptions A and B are satisfied. However, Assumptions A and B are weaker than exchangeability, and depend on the scoring rule $S_L$ being used. In particular, for mean forecasting under squared error loss, Assumptions A and B depend only on the means $\mu_1, \ldots, \mu_n$ of the forecast distributions, so the distributions are allowed to differ regarding other facets such as variance or skewness. For other scoring rules such as CRPS, however, such differences would typically violate Assumptions A and B.

Assumption A is strong in principle but can well be plausible in practice. In particular, there is often uncertainty about the relative forecasting performance of various methods. This can be due to a short history of forecasts and observations (e.g. when forecasting European macroeconomic data, which are available only quarterly since 2002), due to entry and exit of forecasters in surveys or collaborative forecasting efforts \citep[see e.g.][in epidemiology]{CramerEtAl2022}, or due to temporal instability of the methods' relative performance \citep{GiacominiRossi2010}. In the presence of such uncertainty, the equal-performance hypothesis from Assumption A is arguably a plausible default. Assumption B imposes another symmetry-type condition on the methods. The plausibility of this condition hinges on the set of methods being combined. When $n = 2$, the condition is trivially satisfied. Remark \ref{prop:link} links Assumptions A and B to a well-known set of assumptions that ensures optimality of equal weights in the case of mean forecasts \citep[Section 2.4]{Timmermann2006}. 

\begin{rem}\label{prop:link}
Consider mean forecasting under the squared error scoring rule, i.e. $S_L(F,y) = (y-\mu)^2,$ where $\mu = \int z~dF(z)$. Assume further that the $n$ mean forecasts to be combined are all unconditionally unbiased, such that $\mathbb{E}_\mathbb{Q}(Y) = \mathbb{E}_\mathbb{Q}(\mu_i), i = 1, \ldots,n$. Then Assumptions A and B of Proposition \ref{prop:equidist} are equivalent to assuming that the forecast errors $e_i = Y-\mu_i$ have identical variances and pairwise identical correlations.
\end{rem}

\cite{Elliott2025} and \cite{DieboldEtAl2025} argue that the assumptions of Remark \ref{prop:link} are empirically plausible for survey forecasts of US macroeconomic variables.

We next consider weights that are not constant but reflect conditioning information used in the forecasting process. For example, the information set $\mathcal{A}_F = \bigcup_{i=1}^n \sigma(F_i)$ is generated by knowledge of the forecasts $F_1, \ldots, F_n$. As another example, let $\mathcal{A}_i$ denote the information set underlying forecast $i$. Then $\mathcal{A}_{1:n} = \bigcup_{i=1}^n \sigma(\mathcal{A}_i)$ is the union of the information sets underying the individual forecasts. Clearly, $\mathcal{A}_F \subseteq \mathcal{A}_{1:n}$. 

Let $w_\mathcal{A}$ denote the combination weights that optimize expected performance conditional on $\mathcal{A}$, i.e. 
\begin{equation}
w_\mathcal{A}^* = \underset{w \in \text{PS}^n}{\text{arg min}} \left( \sum_{i=1}^n w_i \mathbb{E}_\mathbb{Q}(S_i|\mathcal{A}) - \sum_{i=1}^{n-1} \sum_{j>i} w_iw_j\mathbb{E}_\mathbb{Q}(d_{ij}|\mathcal{A})\right).\label{dynweights}
\end{equation} Constant weights as considered at (\ref{constweights}) arise as a special case when $\mathcal{A} = \emptyset$ is the empty information set.

\begin{prop} Let $w_\mathcal{A}^*$ and $w_\mathcal{B}^*$ denote conditionally optimal weights as described above. If $\mathcal{B} \subseteq \mathcal{A},$ then $w_\mathcal{A}*$ attains a better unconditionally expected score than $w_\mathcal{B}*$, i.e. 
$$\mathbb{E}_\mathbb{Q}\left(\sum_{i=1}^n w_{\mathcal{A}, i}^* S_i - \sum_{i=1}^{n-1} \sum_{j>i} w_{\mathcal{A}, i}^*w_{\mathcal{A}, j}^*d_{ij}\right) \le \mathbb{E}_\mathbb{Q}\left(\sum_{i=1}^n w_{\mathcal{B}, i}^* S_i - \sum_{i=1}^{n-1} \sum_{j>i}^n w_{\mathcal{B}, i}^*w_{\mathcal{B}, j}^*d_{ij}\right).$$ \label{prop:dynamic}
\end{prop}
The proposition states that using a larger information set is necessarily beneficial when constructing forecast combination weights. Importantly, this population-level result assumes that $w_\mathcal{A}^*$ and $w_\mathcal{B}^*$ make optimal use of their respective information sets. The proposition is similar in spirit to results on nested information sets in forecasting (see e.g. \citealt[Corollary 2]{HolzmannEulert2014} or \citealt[Proposition 3.2]{KruegerZiegel2021}). 

Apart from the empty information set, the information set $\mathcal{A}_{F}$ seems particularly relevant from a practical perspective: Since the individual forecast distributions must be known for constructing any combination, they might as well be used for constructing the combination weights. Given that knowledge of 
the forecasts implies knowledge of pairwise divergences (i.e., $d_{ij} \in \mathcal{A}_{F}$ for all $i, j$), the 
optimal weights based on this information set are
\begin{equation}
w_{\mathcal{A}_{F}}^* = \underset{w \in \text{PS}^n}{\text{arg min}} \left( \sum_{i=1}^n w_i \mathbb{E}_\mathbb{Q}(S_i|\mathcal{A}_{F}) - \sum_{i=1}^{n-1} \sum_{j>i} w_iw_j d_{ij}\right).\label{dynweights2}
\end{equation}
The expression $\mathbb{E}_\mathbb{Q}(S_i|\mathcal{A}_{F})$ is a rather complicated object, as it implicitly depends on the distribution of $(F_1, F_2, \ldots, F_n, Y)$. Thus, any attempt to use $w_{\mathcal{A}_F}^*$ in practice requires to either estimate the term (e.g. via regression methods) or to restrict it by means of assumptions.

\begin{rem}
Estimating $\mathbb{E}_\mathbb{Q}(S_i|\mathcal{A}_{F})$ empirically requires to choose regressors that are measurable with respect to $\mathcal{A}_{F}$. In practice, a plausible strategy is to choose a pool of candidates, combined with a prediction method that carefully exploits this pool while avoiding overfitting. Below we pursue this principle, using the pairwise divergences between all forecasts, $(d_{ij})_{j>i}$, as well as the individual forecasts' entropies, $(\text{ent}(F_i))_{i=1}^n$, as regressors. This pool of regressors can be motivated as follows: Based on the definition of divergence and entropy, the conditionally expected score in question can be written as $$\mathbb{E}_\mathbb{Q}(S_i|\mathcal{A}_{F}) = d(F_i, \mathcal{L}(Y|\mathcal{A}_F)) + \text{ent}(\mathcal{L}(Y|A_F)),$$ where $\mathcal{L}(Y|A_F)$ is the true distribution of $Y$ given $\mathcal{A}_F$. The latter distribution is of course unknown in practice. However, if $\mathcal{L}(Y|A_F) = \sum_{i=1}^n w_i F_i$ for some choice of combination weights $w_1,\ldots, w_n$, calculations similar to the proof of Proposition \ref{prop0} imply that $\mathbb{E}_\mathbb{Q}(S_i|\mathcal{A}_{F})$ is a function of the pool of regressors mentioned above. To describe the approach in practice, let $S_{i,t}$ denote the score of method $i$ at date $t$. We seek to model the conditional mean of $S_{i,t}$ as a function of the $0.5 \cdot n \cdot (n-1)$ regressors given by $(d_{ij,t})_{j > i},$ where $d_{ij,t}$ is the divergence between forecasts $i$ and $j$ at date $t$, and the entropy terms for the $n$ distributions at date $t$, $(\text{ent}(F_{i,t}))_{i=1}^n$. For this purpose, we use complete subset regressions (CSRs) due to \cite{ElliottEtAl2013}. We use one regressor at a time ($k = 1$ in their notation), so that the predicted value is a simple average of $0.5 \cdot n \cdot (n+1)$ linear regressions, each of which uses an intercept term in addition to the regressor in question. \label{remark:csr}
\end{rem}

\begin{rem}
Under additional assumptions, Equation (\ref{dynweights2}) can be used to motivate a trimming-type combination of probabilistic forecasts considered by \cite{Taylor2025b}. Trimming is based on the notion that a forecast is more promising if it is more central, where centrality can be measured by means of statistical depth functions \citep{Mosler2022}. Let $\text{depth}_i$ denote the depth of forecast $i$, with larger values indicating that the forecast is more central. Since the depth of each forecast can be determined on the basis of $F_1, F_2, \ldots, F_n$, it holds that $\text{depth}_i \in \mathcal{A}_F$ for all $i$.
\cite{Taylor2025b} considers using a weight of one for the deepest forecast, a method he labels `Deepest'. This method can be motivated from Equation (\ref{dynweights2}), together with a prior belief that the deepest forecast outperforms the other forecasts by a sufficient amount. For example, assume that $\mathbb{E}_\mathbb{Q}(S_i|\mathcal{A}_{F}) = \theta_0 - \theta_1 \text{depth}_i,$ where $\theta_0$ and $\theta_1 > 0$ are parameters that reflect the prior beliefs of the person choosing the combination weights. If $\theta_1$ is sufficiently large, then $w_{\mathcal{A}_F}$ places all weight on the forecast $i^* \in \{1, 2, \ldots, n\}$ that maximizes $\text{depth}_i$.
\end{rem}

\subsection{(Lack of) Calibration of the Linear Pool}

\begin{definition}
For a proper scoring rule $S$, let $$\Delta_S= \mathbb{E}_\mathbb{Q}[S(F_i, Y)] - \mathbb{E}_{\mathbb{Q}} [\text{ent}(F_i)].$$ The forecast distribution $F_i$ is called realistic if $\Delta_S = 0$. $F_i$ is called underconfident if $\Delta_S < 0$, and overconfident if $\Delta_S > 0$. \label{ass:realism}
\end{definition}

The term $\Delta_S$ from Definition \ref{ass:realism} considers the difference between a forecast's actual performance and its entropy (i.e., the estimate of its own performance). This difference has been studied by \cite{KnueppelEtAl2022} and others. Definition \ref{ass:realism} requires the forecast $F_i$ to be realistic in an `average' sense that is weaker than auto-calibration \citep[Proposition 2]{Krueger2024}.

\begin{prop}
Suppose that a kernel scoring rule is used, and that $n$ forecast distributions are all realistic. Furthermore, consider a linear pool with fixed weights $w \in \text{PS}^n$. Then the linear pool is underconfident, and the degree of underconfidence is given by $2~\mathbb{E}_Q(D) \ge 0$. \label{prop:under}
\end{prop}

Proposition \ref{prop:under} generalizes \citet[Proposition 1]{KnueppelKrueger2022} from squared error and a univariate predictand to all kernel scores and a possibly multivariate predictand. Intuitively, its result is due to a dual role of disagreement $D$: On the one hand, the pool's entropy increases as $D$ increases (Proposition \ref{prop0}), i.e. the forecast distribution becomes less confident. On the other hand, ceteris paribus $D$ improves the pool's realized performance (Proposition \ref{prop1}). The proposition's broader implication that the combination of realistic components is underconfident has been noted by \cite{Hora2004} for a continuous univariate predictand, by \cite{RanjanGneiting2010} for a binary univariate predictand, and by \cite{GneitingRanjan2013} for a general univariate predictand, using various notions of calibration. We are not aware of results for a general multivariate predictand as covered by Proposition \ref{prop:under}.

Proposition \ref{lpvsqp} presents conditions under which quantile pooling is more effective than linear pooling, in the sense that it attains smaller entropy (i.e., is more confident) without introducing overconfidence. 

\begin{prop}
Consider the case of a univariate predictand with support $\Omega \subseteq \mathbb{R}$, and assume that 
\begin{itemize}
\item[A] $F_1, F_2, \ldots, F_n$ are realistic.
\item[B] $F_1, F_2, \ldots, F_n$ are strictly increasing on some interval $\Omega_F \subseteq \mathbb{R}$, $\mathbb{Q}$-almost surely.
\item[C] There is an interval $(\alpha_1, \alpha_2) \subseteq [0,1]$ such that $$\mathbb{P}_\mathbb{Q}\left(\text{min} \{F_i^{-1}(\alpha)\}_{i: w_i > 0} < Y < \text{max} \{F_i^{-1}(\alpha)\}_{i: w_i > 0}\right) > 0$$ for each $\alpha \in (\alpha_1, \alpha_2).$
\end{itemize}
Suppose that the CRPS is used, and consider the linear pool $F_w$ and quantile combination $F_w^q$ using the same fixed weight vector $w \in \text{PS}^n$. Then the quantile combination is underconfident and satisfies $\mathbb{E}_{\mathbb{Q}}[\text{ent}(F_w^q)] \le \mathbb{E}_{\mathbb{Q}}[\text{ent}(F_w)]$. \label{lpvsqp}
\end{prop}

Assumption B simplifies the setup by ensuring that quantiles are unique. Assumption C ensures that there is a strictly positive probability of `bracketing', i.e. the event that the realization falls between two predicted quantiles \citep[c.f.][]{GrushkaEtAl2017}. This assumption is mild and merely rules out very peculiar setups (such as the case of $n$ identical quantile curves, or situations where the quantile curves are guaranteed to collectively under- or overpredict $Y$ at all quantile levels). 

Proposition \ref{lpvsqp} implies that quantile-based combination has better calibration properties than linear pooling. That said, the following result states that the calibration properties of quantile-based combination are not satisfactory either.

\begin{prop}
Consider $n$ auto-calibrated forecast distributions, and suppose that Assumptions B and C of Proposition \ref{lpvsqp} hold. Then quantile-based combination fails to produce an auto-calibrated forecast distribution. \label{prop:qpac}
\end{prop}

This negative result on quantile-based combination mirrors the existing results for linear pooling discussed above. The only related `calibration failure' result for quantile combination we are aware of is \citet[Theorem 5]{Tibshirani2023}, which shows that marginal calibration \citep[see][Section 2.2]{GneitingResin2023} is not preserved under quantile combination. See Section \ref{sec:lcalib} for a simple illustration of Proposition \ref{prop:qpac}.

\section{Gaussian Example}

\label{sec:simulation}

Here we present a Gaussian example that illustrates the results on weighting schemes and underonfidence discussed in Section \ref{sec:implications}. 

\subsection{Setup}

Let $Y = X_1 + X_2 + X_3 + \varepsilon,$ with 
$$X = \begin{pmatrix}
X_1, & X_2, & X_3 
\end{pmatrix}' \sim \mathcal{N}\left(\mathbf{0}, \begin{pmatrix}
1 & 0 & \rho \sigma_3 \\ 0 & 1 & \rho \sigma_3 \\ \rho \sigma_3 & \rho \sigma_3 & \sigma_3^2 \end{pmatrix}\right),$$ where $\sigma_3 > 0$ and $\rho \in (-\frac{1}{\sqrt{2}}, \frac{1}{\sqrt{2}})$ are parameters and $Z$ is a Bernoulli random variable with success probability $1/2$ that is independent of $X$. The restriction on $\rho$ ensures that the covariance matrix of $X$ is positive definite. Conditional on $Z = z$, $\varepsilon \sim \mathcal{N}(0, 1 + z).$ 

We consider a setting with three forecasters, where forecaster $i \in \{1, 2, 3\}$ issues the correct conditional distribution of $Y$ given $Z = z$ and $X_i = x_i$. Since $\begin{pmatrix}
X_i, &Y \end{pmatrix}$ is bivariate normal given $Z = z$, the conditional distributions of $Y|Z,X_i$ is also normal. For $i = 1,2$, we have
\begin{eqnarray*}
\mathbb{E}_{\mathbb{Q}}(Y|Z=z,X_i = x_i) &=& x_i \cdot (1 + \rho \sigma_3) = M_i,\\
\mathbb{V}_{\mathbb{Q}}(Y|Z=z,X_i = x_i) &=& 2 + z + \sigma_3^2(1-\rho^2) + 2 \rho \sigma_3 = V_i.
\end{eqnarray*}
Furthermore, 
\begin{eqnarray*}
\mathbb{E}_{\mathbb{Q}}(Y|Z=z, X_3 = x_3) &=& x_3 \cdot (1 + 2\rho/\sigma_3) = M_3\\
\mathbb{V}_{\mathbb{Q}}(Y|Z=z, X_3 = x_3) &=& 3 + z - 4 \rho^2 = V_3
\end{eqnarray*}

The setting implies that Forecasts 1 and 2 are exchangeable. If $\sigma_3$ and $|\rho|$ are sufficiently large, then $X_3$ is a more informative forecast than $X_1$ or $X_2$, performing better in terms of expected squared error or CRPS: If $\sigma_3$ is large, then $X_3$ is an important component of $Y$, so having access to it is particularly helpful. If $|\rho|$ is large, then $X_3$ is helpful for predicting $X_1$, $X_2$ and thus $Y$. The setting further implies that all individual forecast distributions are realistic (see Definition \ref{ass:realism}). Hence by Proposition \ref{prop:under}, any non-trivial linear pool will be underconfident.

\subsection{Weighting Schemes}

We consider the following weighting schemes. First, equal weights of $1/3$ for each forecast. Second, static weights that are the finite sample analogue of  $w^*$ defined at (\ref{constweights}). Third, weights $w_{\mathcal{A}_F}^*$ as defined at (\ref{dynweights2}). As discussed there, the weights are infeasible in practice, which is why we refer to them as `dynamic oracle' weights. The weights can be computed analytically in the current setup, and are informative as a hypothetical performance benchmark. We next derive an expression for $\mathbb{E}_{\mathbb{Q}}(S_i|\mathcal{A}_F)$ which is required to compute the weights. Let $\mu_\mathbf{X}$ and $\sigma^2_\mathbf{X}$ denote the mean and variance of $Y|X_1,X_2,X_3,Z$, and note that $\mathbb{E}_\mathbb{Q}(Y|\mathcal{A}_F) = \mu_\mathbf{X} = X_1 + X_2 + X_3$ and $\mathbb{V}_\mathbb{Q}(Y|\mathcal{A}_F) = \sigma^2_\mathbf{X} = 1 + Z$. Then $$\mathbb{E}_{\mathbb{Q}}(S_i|\mathcal{A}_F) = d_\mathcal{N}(\mu_\mathbf{X},\sigma^2_\mathbf{X}, M_i, V_i) + \text{ent}_\mathcal{N}(\mu_\mathbf{X},\sigma^2_\mathbf{X}),$$ 
where $d_\mathcal{N}$ and $\text{ent}_\mathcal{N}$ are the CRPS divergence and entropy functions in the Gaussian case (see Appendix \ref{crps_gaussian}). Fourth, a feasible counterpart of $w_{\mathcal{A}_F}$ using complete subset regressions, as decribed in Remark \ref{remark:csr}. We refer to this method as `dynamic'.

For comparison to the various linear pooling methods, we also consider an equally weighted quantile combination. As an infeasible theoretical benchmark, we further consider the true forecast distribution $\mathcal{L}(Y|\mathcal{A}_F)$ which combines the information sets underlying the individual forecasts (i.e., $X_1, X_2$ and $X_3$) instead of the forecasts themselves. In all three scenarios, this approach yields a Gaussian distribution with mean $X_1 + X_2 + X_3$ and variance $1+Z$. Since $Z$ is a Bernoulli random variable with success probability one half, the expected CRPS of $\mathcal{L}(Y|\mathcal{A}_F)$ is $0.5 \times (1 + \sqrt{2})/\sqrt{\pi} \approx 0.68$ (see Appendix \ref{crps_gaussian}).\footnote{While also infeasible, the `dynamic oracle' method is forced to use a linear pool of the individual forecasts. This form is misspecified since the true distribution $\mathcal{L}(Y|\mathcal{A}_F)$ is Gaussian, so that the `dynamic oracle' performs worse than $\mathcal{L}(Y|\mathcal{A}_F)$ .}

\subsection{Simulation Results}

Here we present simulation results for $1\,000$ Monte Carlo iterations. In each iteration, we use a training sample of $n_{\text{train}}$ observations (forecasts and outcomes) for estimating the weights, followed by a test sample of $100$ observations. We thus obtain a total of $100\,000$ test sample observations. For each weighting scheme, we compute the average CRPS across these observations. We set $n_{\text{train}}$ to either $16$ or $125$, corresponding to the values used in our empirical case studies below. We consider three scenarios:

\subsubsection*{Scenario 1: Exchangeable Forecasts}
We first set $\rho = 0$ and $\sigma_3 = 1$, so that all three forecasts are exchangeable, and the conditions of Proposition \ref{prop:equidist} are satisfied. Hence equal weights are optimal.

\subsubsection*{Scenario 2: Moderate Differences in Forecast Performance}

Here we set $\rho = 0.2$ and $\sigma_3 = 1.2,$ so that the expected CRPS of forecast $3$ is about $12\%$ smaller (i.e., better) than the expected CRPS of forecasts 1 and 2.

\subsubsection*{Scenario 3: Pronounced Differences in Forecast Performance}

Here we set $\rho = 0.25$ and $\sigma_3 = 1.8$, so that forecast $3$ is about $30\%$ more accurate than the other forecasts in terms of expected CRPS.

\begin{table}
\centering
\begin{tabular}{lrrrrrr}
\toprule
&  \multicolumn{2}{c}{Scenario 1} & \multicolumn{2}{c}{Scenario 2} & \multicolumn{2}{c}{Scenario 3}\\
& $n = 16$ & $n = 125$ & $n = 16$ & $n = 125$ & $n = 16$ & $n = 125$ \\
\midrule
forecast3 & 1.05 & 1.05 & 1.03 & 1.03 & 1.01 & 1.01\\ \addlinespace
static & 0.99 & 0.96 & 1.01 & 0.97 & 1.01 & 0.98\\
equal & 0.96 & 0.96 & 0.99 & 0.98 & 1.09 & 1.09\\
dynamic oracle & 0.81 & 0.80 & 0.81 & 0.81 & 0.84 & 0.83\\
dynamic & 0.98 & 0.97 & 1.00 & 0.98 & 1.03 & 1.01\\
\addlinespace
equal quantile comb. & 0.95 & 0.95 & 0.97 & 0.97 & 1.08 & 1.08\\
\bottomrule
\end{tabular}
\caption{Simulation results for the design described in Section \ref{sec:simulation}. Forecast 3 is the best individual forecast in all settings. The other rows represent combination methods as described in the text. The expected CRPS of the true forecast distribution is $0.68$ in all cases. \label{tab:sim}}
\end{table}

\subsubsection*{Results}

Table \ref{tab:sim} presents the simulation results. While the equal and `dynamic oracle' methods use fixed and known parameters, the `static' and `dynamic' methods feature finite sample estimation uncertainty. Comparisons of their performance across sample sizes ($n = 16$ versus $n = 125$) reveal the expected pattern that a larger sample size is beneficial. That said, this effect is quantitatively rather small, relative to variation of CRPS scores across methods (rows). Similarly, the performance gap between static and equal weights is small in Scenario 1 where equal weights are optimal.

Static and dynamic weights generally perform similarly, and dynamic oracle weights clearly outperform both of them. This indicates a potential for further benefits of the dynamic approach when using the regressor pool defined by $\mathcal{A}_F$ more effectively. Empirical evidence seems necessary to clarify how much of these potential gains can realistically be realized in practice.

The combinations mostly outperform forecast 3, which is the best individual forecast in all settings. This is particularly interesting in the case of static and equal weights, which generally yield underconfident forecasts (Proposition \ref{prop:under}), whereas forecast 3 is realistic. In Scenario 1, equal weights are even maximally underconfident (as they maximize $\mathbb{E}_\mathbb{Q}(D)$) but are nevertheless optimal in terms of expected CRPS. 

We further find small but consistent gains of equally weighted quantile combination as compared to equally weighted linear combination. These gains can be explained by the fact that quantile combination yields a sharper forecast distribution, which is desirable in the present case (see Proposition \ref{lpvsqp}). 

\subsection{Lack of Calibration} \label{sec:lcalib}

Propositions \ref{prop:under} and \ref{prop:qpac} imply that both linear and quantile-based combination fail to be auto-calibrated in the present setup. Here we provide a simple illustration of this result, focusing on Scenario 1 (exchangeable forecasts) and equally weighted combination for brevity. In this setup, the linear pool is an equally weighted mixture of three distributions with respective mean $X_i$ and variance $3+Z$. Hence the variance of the combined distribution is $3 + Z + 1/9~\sum_{i=1}^2\sum_{j > i} (X_i-X_j)^2,$ with an expected value of $25/6 \approx 4.16$. The quantile-based combination is a Gaussian distribution with a variance of $3 + Z$, i.e. an expected variance of $3.5$. The expected squared error of $\bar{X} = \frac{1}{3} \sum_{i=1}^n X_i$, which is the mean forecast of both combination methods, is $17/6 \approx 2.83$, i.e. less than the expected variance of either combination method. Hence neither combination yields an auto-calibrated forecast distribution, for which the expected entropy (e.g. variance) would coincide with the expected score (e.g. squared error).

\section{Empirical Case Studies}\label{sec:empirical}

This section provides empirical illustrations on inflation and exchange rate forecasting. 

\subsection{Bivariate Forecasting of Two Inflation Measures}\label{sec:illustration}

We first consider bivariate forecasting of two popular US inflation measures, based on the consumer price index (CPI) and the price index of GDP (PGDP). For each measure, we consider quarterly annualized growth rates of the underlying index. Figure \ref{fig:inflation_actuals} in the appendix plots the two time series. The series are highly correlated, so that either series is of potential help for predicting the other. We construct bivariate forecast distributions using two methods. First, we use average point forecasts from the \citeauthor{PhiladelphiaFed2024a}'s Survey of Professional Forecasters (SPF). We use the bivariate empirical distribution of the SPF's historical forecast errors to construct a forecast distribution. This is a simple bivariate `postprocessing' method based on the principle of using past forecast errors in order to estimate future forecast uncertainty; see \cite{SchefzikEtAl2013} for further discussion. As a second forecasting method, we use a Bayesian vector autoregressive (BVAR) model with stochastic volatility. Specifically, we use the model proposed by \cite{Primiceri2005} and \cite{delNegro2015}, as implemented in the R package \textsf{bvarsv} \citep{Krueger2015}. We employ the default setting of the latter implementation -- in particular, using a single autoregressive lag, and priors that allow for time variation in both the mean and variance equations of the BVAR. For constructing the forecast distributions, we employ real-time data on CPI and PGDP, as provided by \cite{PhiladelphiaFed2024b}. We use an expanding estimation window, with data ranging back until 1980. We consider both current-quarter forecasts and one-year-ahead forecasts. Since the current quarter's observation is not yet available to SPF participants, these two horizons correspond to $h = 1$ and $h = 5$ quarter ahead forecasts. We consider forecasts made between 1994:Q4 (the earliest quarter for which the Philadelphia Fed's real-time data is available) and 2024:Q2.

We use the Energy Score for evaluating the bivariate probabilistic forecasts; see Appendix \ref{sec:specific} for details. At any given date and forecast horizon, the two component distributions $F_1$ (SPF) and $F_2$ (BVAR) are equally weighted empirical distributions of $40$ and $5\,000$ observations respectively. For linear pooling, we consider equal, static or dynamic combination weights, where the latter are constructed as described in Remark \ref{remark:csr}. We use a rolling window of $16$ observations for estimating the static and dynamic weights.\footnote{The label `static' refers to the fact that these weights are a finite sample analogue of the unconditional weights at (\ref{constweights}). In practice, the use of a rolling estimation window introduces time variation into these weights.}

Figure \ref{fig:es_spf_bvar_h1} in the appendix plots the equally weighted pool's expected ES (i.e., its entropy) and its components for current-quarter forecasts. In most quarters, disagreement accounts for a modest share of the expected ES, with an average share of $7\%$. Two notable exceptions with large disagreement arise in 2008:Q4 and 2020:Q2. These two quarters are associated with the global financial crisis and the Covid-19 pandemic respectively. In these quarters, disagreement peaks both in absolute terms and regarding its share among the linear pool's expected score ($46.6\%$ in 2008:Q4, and $52.4\%$ in 2020:Q2).
The top row of Figure \ref{fig:examples} shows the bivariate forecast distributions for these quarters. In both instances, the SPF distribution is located to the southwest of the BVAR distribution, indicating lower inflation rates according to both measures (CPI and PGDP). The SPF's assessment is in line with the eventual realizations, and can be explained by the survey's access to more timely intra-quarter information that is not available to the BVAR. In particular, the SPF point forecasts correctly anticipate the dis-inflationary short-term impact of the economic shocks of 2008:Q4 and 2020:Q2. Thus, the SPF's access to recent and possibly judgmental information is beneficial in these examples of short-term forecasts in turbulent periods. 

Interestingly, the effect just described is not present for one-year-ahead forecasts ($h = 5$). Information about the current state of the economy seems to matter less at this longer horizon, where it is dominated by shocks occurring between the forecast date and the target date \citep[c.f.][Section 4.4]{KruegerEtAl2017}. To illustrate this point, the bottom row of Figure \ref{fig:examples} shows the one-year-ahead forecast distributions in 2008:Q4 and 2020:Q2. Disagreement between the SPF and BVAR distributions is small even in these turbulent periods. More broadly, at $h = 5$ the share of disagreement among the linear pool's entropy is $3\%$ on average, with a maximal share of $16\%$ attained in 2001:Q2. Figure \ref{fig:es_spf_bvar_h5} in the appendix shows details for $h = 5$.

Table \ref{tab:scores_spf_bvar} summarizes the forecast performance of the SPF and BVAR forecasts as well as the three linear pools. Among the two individual forecasts, the SPF attains a better score at $h = 1$ while the BVAR prevails at $h = 5$.
Among the three pools, equal weights perform worst at $h = 1$ but best at $h = 5$. Static and dynamic weights generally perform similarly. 

In Table \ref{tab:dm}, we consider \cite{DieboldMariano1995} tests of equal predictive ability between any pair of pooling methods. To estimate the variance in the relevant $t$ statistic, we use the function \verb|NeweyWest| in the \verb|R| package \textsf{sandwich} \citep{Zeileis2004,ZeileisEtAl2020} which allows for autocorrelation. Using a two-sided test at the 5\% level, pairwise differences in performance are statistically insignificant. The test results under squared error loss, which we report in Table \ref{tab:dm_se}, are qualitatively identical.

As discussed in Section \ref{sec:implications}, comparing a method's expected and actual scores provides evidence on calibration. To test the null that a method is realistic about its own performance (see Definition \ref{ass:realism}), we again use a $t$ test.  \cite{KnueppelEtAl2022} refer to this method as an `entropy test'. Using the same variance estimator as for Diebold-Mariano testing, we do not reject the null in any of the six cases considered (three methods times two horizons), yielding no evidence against calibration. See Table \ref{tab:calibration} for details. 

\begin{table}
\centering
\begin{tabular}{lrr}
\toprule
 & $h =  1$ & $h =  5$\\
\midrule
SPF & 1.10 & 1.83\\
BVAR & 1.47 & 1.67\\
LP (equal weights) & 1.19 & 1.70\\
LP (static weights) & 1.10 & 1.76\\
LP (dynamic weights) & 1.11 & 1.74\\
\bottomrule
\end{tabular}

	\caption{Energy Score for bivariate forecast distributions of CPI and PGDP inflation. The evaluation sample covers forecasts made from 1998:Q4 onwards. The latest observation refers to forecasts made in 2024:Q2 (for $h = 1$, resulting in 103 observations) or 2023:Q2 (for $h = 5$, resulting in 95 observations). Realizations are computed based on second-vintage data.  \label{tab:scores_spf_bvar}}
\end{table}

\begin{figure}
	\begin{tabular}{cc}
		2008:Q4, $h = 1$ & 2020:Q2, $h = 1$ \\[-.7cm]
		\includegraphics[width=.5\textwidth]{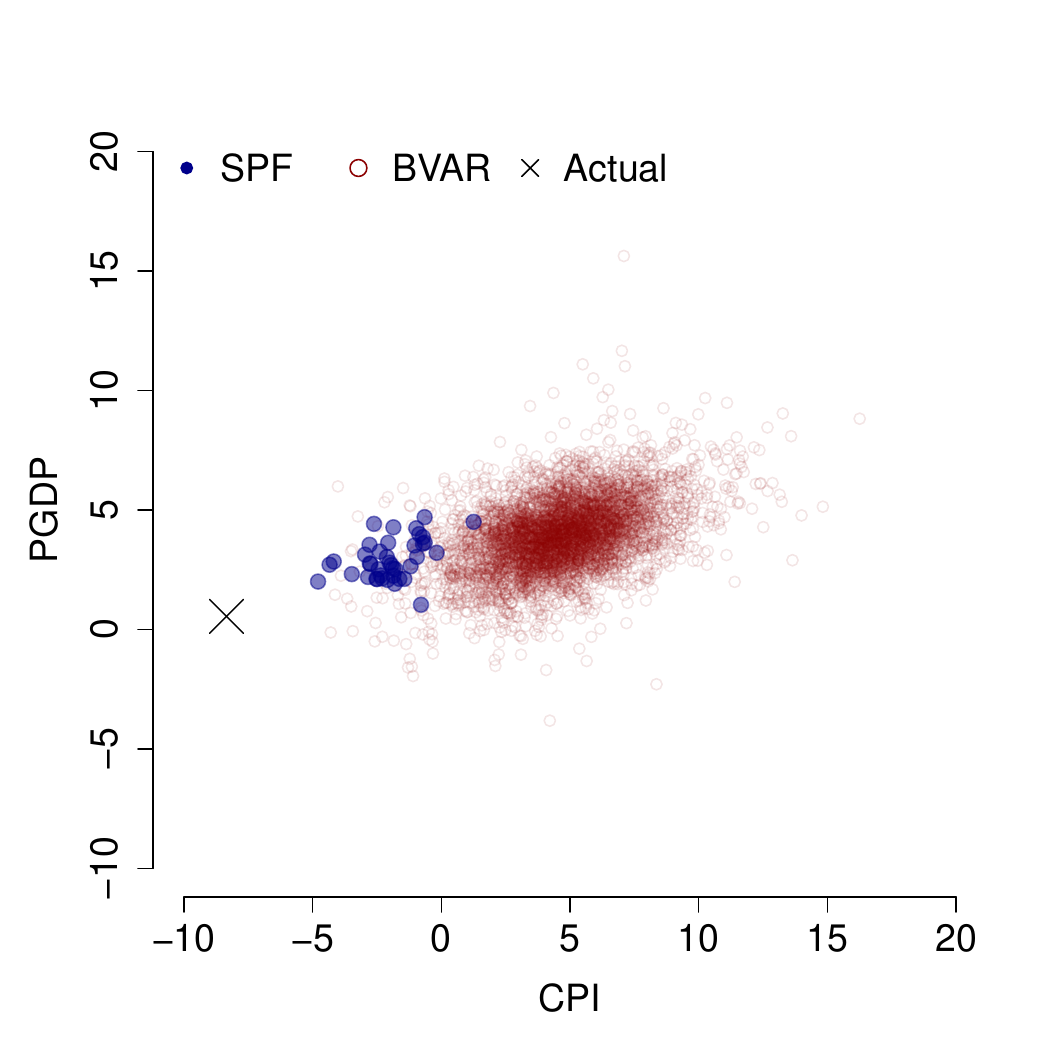}	&
		\includegraphics[width=.5\textwidth]{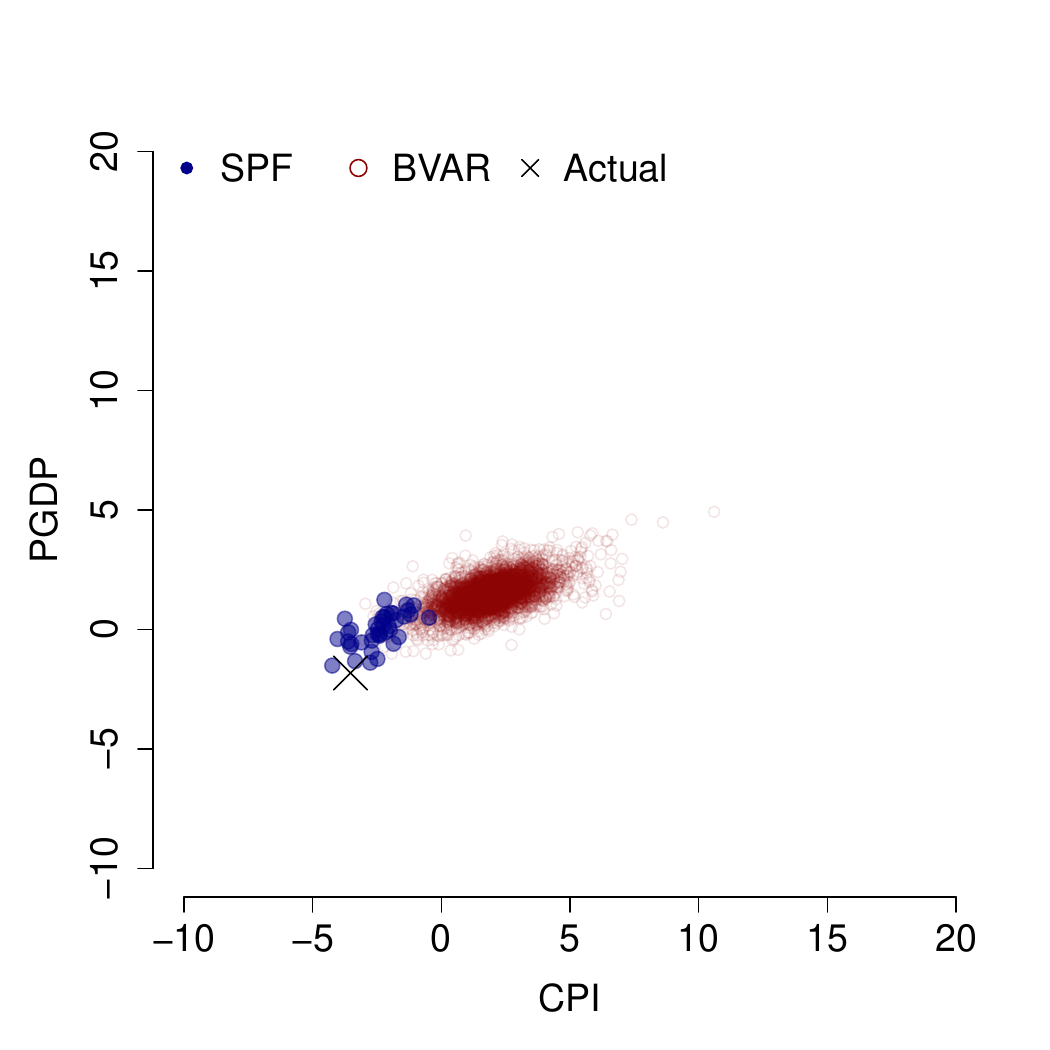}\\[.5cm]
		2008:Q4, $h = 5$ & 2020:Q2, $h = 5$ \\[-.7cm]
		\includegraphics[width=.5\textwidth]{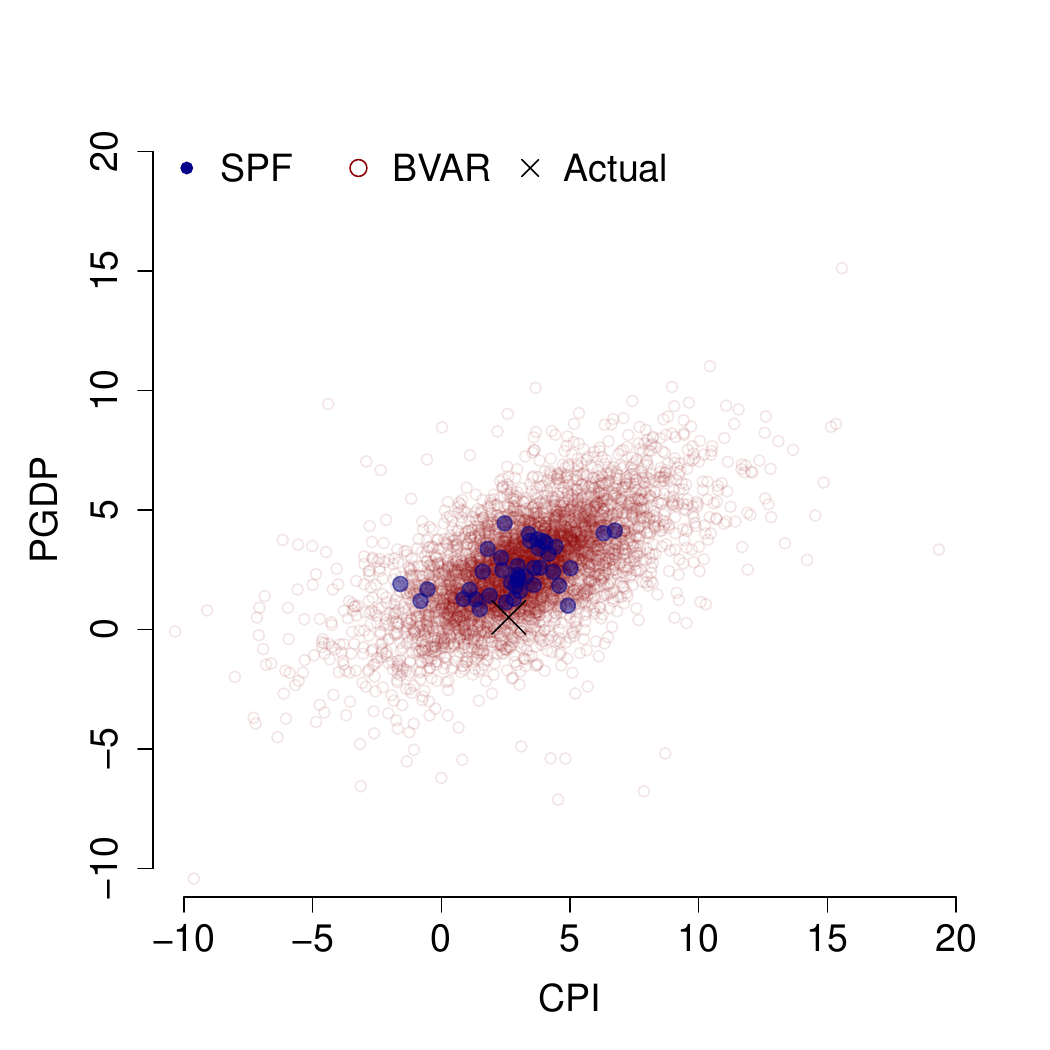}	&
		\includegraphics[width=.5\textwidth]{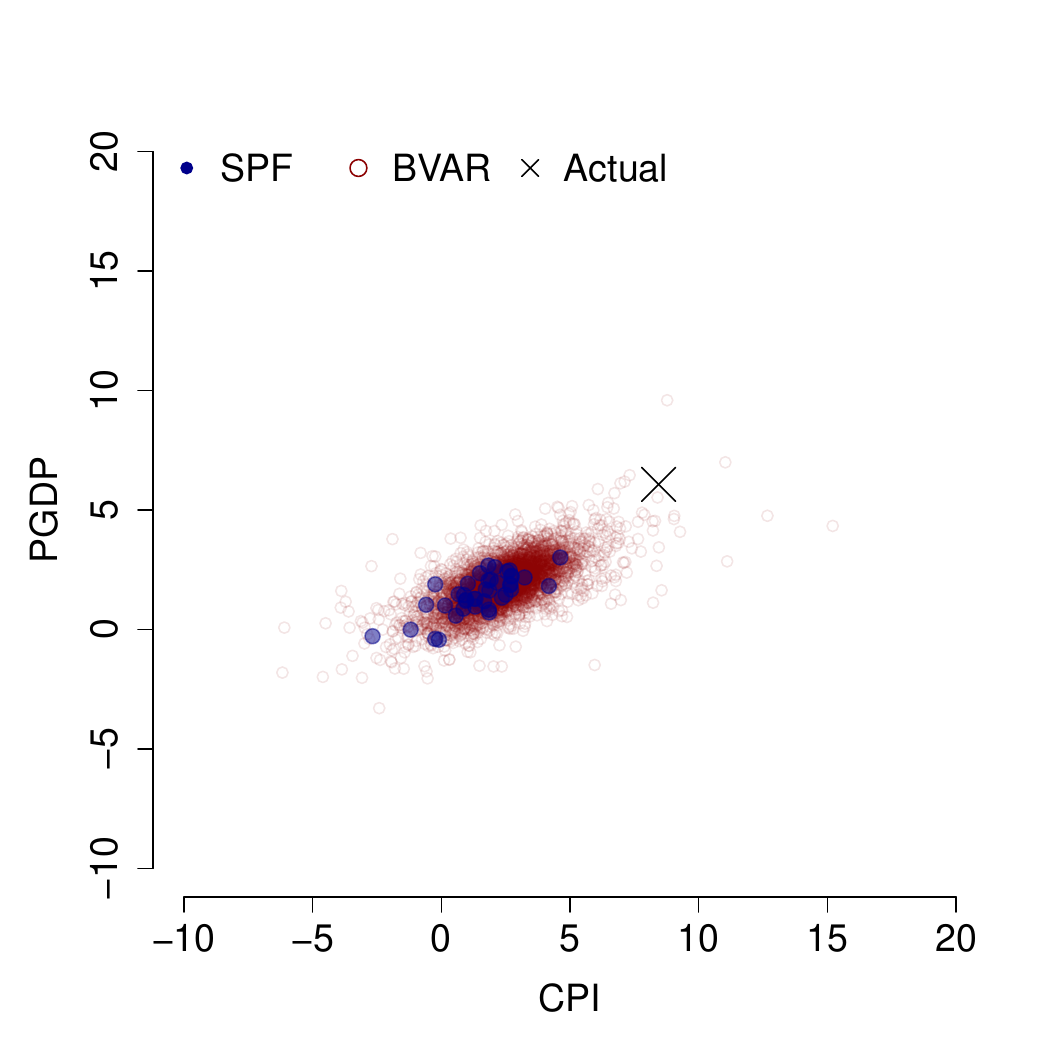}\\
	\end{tabular}
	\caption{Bivariate forecast distributions in 2008:Q1 and 2020:Q2. For $h = 1$, these are the two quarters in which disagreement is maximal (see Figure \ref{fig:es_spf_bvar_h1}). \label{fig:examples}}
\end{figure}

\subsection{Exchange Rate Forecasting}

Our second case study considers distributional forecasting of exchange rates. Making point forecasts of exchange rates is notoriously difficult in the sense that simple benchmarks are hard to beat (see \citealt{Rossi2013} and the references therein). Predicting other facets of the conditional distribution of exchange rates, such as their conditional variance, seems more promising and is also practically relevant for risk management purposes. A large number of studies, including \cite{HansenLunde2005}, consider predicting exchange rate volatility, often via generalized autoregressive heteroscedasticity (GARCH) type models. A smaller number of papers including \cite{AbbateMarcellino2018} and \cite{Huber2019} consider forecasting the full distribution of exchange rates. 

In order to forecast the exchange rate at date $t$, we simply use the $M$ most recent observations $y_{t-1}, y_{t-2}, \ldots, y_{t-M}$ available at date $t-1$. The window length $M$ is an important tuning parameter: A small choice of $M$ yields a flexible distribution that adapts to changing conditions, while a large choice of $M$ offers more stability and less estimation noise. The choice of $M$ that best governs this trade-off is largely an empirical question. Rather than using a single choice of $M$, we consider forecast combination over a grid of ten plausible choices ranging from $M = 250$ (representing roughly one year of data for each business day) to $M = 2500$ (representing roughly ten years). Similar approaches to averaging over forecasts constructed from different data windows have been studied by \cite{Pesaran2011}, among others.

More specifically, let $M_i$ denote the $i$th choice of window length, where $i = 1, 2, \ldots, n = 10$. A combined forecast distribution with weights $w_1, w_2, \ldots, w_{10}$ then places weight\\ $\sum_{i=1}^{10} w_i~\mathbf{1}(M_i \le j)/M_i$ on observation $y_{t-j}$, where $\mathbf{1}(A)$ is the indicator function of the event $A$. Thus, the weights produced by the combined forecasts are necessarily decreasing in the distance $j$, which seems plausible. Figure \ref{fig:lags} illustrates this setup. As for inflation, we use either equal, static or dynamic combination weights. Static weights are based on a rolling window of 125 observations (i.e., roughly half a year of data for each business day). 

We consider the exchange rate of the US Dollar against the Canadian Dollar, the Euro, the Japanese Yen and the British Pound.\footnote{The data is available from \cite{fred2026}, corresponding to the series codes \textsf{DEXCAUS}, \textsf{DEXUSEU}, \textsf{DEXJPUS} and \textsf{DEXUSUK}.} Table \ref{tab:scores_exch} summarizes the empirical results in terms of CRPS, whereas Table \ref{tab:dm} reports Diebold-Mariano comparisons between pooling methods. Dynamic weights tend to perform best overall, with the difference to equal and dynamic weights being statistically significant for Canada. That said, the results for squared error reported in Table \ref{tab:dm_se} indicate that equal weights tend to outperform the static and dynamic methods. 

As regards calibration, the combined forecasts tend to be underconfident (see Table \ref{tab:calibration}). This effect is statistically significant at the 5\% level for two series (Canada and UK) and all three weighting schemes. 

\begin{table}
\centering
\begin{tabular}{lrrrr}
\toprule
 & CA & EU & JP & UK\\
\midrule
Best individual & 0.1811 & 0.2417 & 0.3446 & 0.2420\\
Worst individual & 0.1842 & 0.2448 & 0.3478 & 0.2440\\ [.2cm]
Linear pool (equal weights) & 0.1824 & 0.2423 & 0.3456 & 0.2424\\
Linear pool (static weights) & 0.1807 & 0.2417 & 0.3457 & 0.2416\\
Linear pool (dynamic weights) & 0.1799 & 0.2419 & 0.3452 & 0.2414\\
\bottomrule
\end{tabular}
	\caption{CRPS for exchange rate forecast distributions based on rolling windows of different lengths $M$. Exchange rates are between the US Dollar and the Canadian Dollar (CA), the Euro (EU), the Japanese Yen (JP) and the British Pound (UK). Evaluation sample ranges from July 12, 2023 to October 24, 2025 ($550$ observations). Best and worst individual methods (= window lengths) are given in Table \ref{tab:bestworst}.\label{tab:scores_exch}}
\end{table}

\begin{table}
	\centering
	\begin{tabular}{lrrrr}
		\toprule
		& CA & EU & JP & UK\\
		\midrule
		Best individual & 250 & 1250 & 750 & 2500\\
		Worst individual & 2500 & 500 & 1750 & 500\\
		\bottomrule
	\end{tabular}
	\caption{Best and worst performing window lengths for the forecasting results in Table \ref{tab:scores_exch}. \label{tab:bestworst}}
	\end{table}

\begin{table}
\centering
\begin{tabular}{lrrr}
\toprule
& Equal-Static & Equal-Dynamic & Static-Dynamic \\ \midrule
& \multicolumn{3}{c}{Inflation} \\
$h = 1$ & 1.54 & 1.92 & -0.34\\
$h = 5$ & -1.51 & -1.00 & 0.47\\ \midrule
& \multicolumn{3}{c}{Exchange Rates} \\
CA & 2.79 & 4.54 & 2.87\\
EU &1.87 & 1.17 & -0.57\\
JP &-0.12 & 0.50 & 0.79\\
UK &1.47 & 1.85 & 0.40\\
\bottomrule
\end{tabular}
\caption{Diebold-Mariano comparisons between weighting methods for the linear pool. For example, the first column considers the null hypothesis that $\mathbb{E}(S_\text{equal} - S_\text{static}) = 0,$ where `equal' and `static' refer to the linear pool using equal or static weights, and $S$ denotes the scoring rule of interest (Energy Score for inflation data, CRPS for exchange rate data). The test statistic is standard normally distributed under the null. A positive test statistic indicates that `static' performs better. See Tables \ref{tab:scores_spf_bvar} and \ref{tab:scores_exch} for further information on the inflation and exchange rate data. \label{tab:dm}}
\end{table}
	
\begin{table}
\centering
\begin{tabular}{lrrr}
\toprule
& Equal-Static & Equal-Dynamic & Static-Dynamic \\ \midrule
& \multicolumn{3}{c}{Inflation} \\
$h = 1$ & 1.55 & 1.76 & -0.35\\
$h = 5$ & -1.14 & -0.62 & 0.60\\\midrule
& \multicolumn{3}{c}{Exchange Rates} \\
CA & -1.71 & -1.33 & 0.11\\
EU &-1.79 & -0.90 & 0.52\\
JP &-1.82 & -2.51 & -1.64\\
UK &-1.14 & -0.99 & -0.34\\
\bottomrule
\end{tabular}
\caption{Like Table \ref{tab:dm}, but using squared error instead of the Energy Score/CRPS as an evaluation criterion. \label{tab:dm_se}}
\end{table}

\begin{table}
\centering
\begin{tabular}{lrrr}
\toprule
& Equal & Static & Dynamic \\ \midrule
& \multicolumn{3}{c}{Inflation} \\
$h = 1$ &-0.01 & 1.42 & 1.20\\
$h = 5$ & 0.51 & 0.26 & 0.42\\
& \multicolumn{3}{c}{Exchange Rates} \\
CA & -6.14 & -2.40 & -3.11\\
EU & -1.14 & -1.65 & -1.51\\
JP &1.84 & -0.09 & 0.59\\
UK &-6.45 & -3.94 & -4.03\\
\bottomrule
\end{tabular}
\caption{$t$ statistics for the null that $\mathbb{E}_\mathbb{Q}(\Delta_S) = 0,$ see Section \ref{ass:realism}. The test statistic is standard normally distributed under the null. A positive test statistic points to overconfidence, a negative test statistic points to underconfidence. See Tables \ref{tab:scores_spf_bvar} and \ref{tab:scores_exch} for further information on the inflation and exchange rate data. \label{tab:calibration}}
\end{table}

\clearpage

\section{Discussion}\label{sec:discussion}

\subsection*{Equal versus Estimated Combination Weights}

Proposition \ref{prop:equidist} offers sufficient conditions under which equal weights are optimal in population. In practice, the relevant question is whether any violations of these conditions are severe enough to justify the costs of weight estimation \citep[c.f.][]{Elliott2025}. These costs are both statistical (in terms of weight estimation error) and conceptual (in terms of additional tuning parameters such as the sample used for weight estimation). In the simulation study of Section \ref{sec:simulation}, estimated weights outperform equal weights even when few observations are available for estimation. In the empirical examples of Section \ref{sec:empirical}, however, estimated weights do not yield clear or consistent benefits over equal weights. While similar results are well documented for point forecast combination under squared error \citep[Section 2.6]{WangEtAl2023}, further empirical results seem necessary to assess the case of distributional forecasting under scoring rules such as the CRPS or energy score. 

\subsection*{Composition of the Linear Pool} 

There is often a prohibitively large set of potential forecast methods, and it is necessary to choose a specific subset of methods that enters a linear pool. This selection is particularly relevant when using equal combination weights, since there is no weight estimation step that can fine-tune the initial selection. Suppose one currently uses an equally weighted pool of $n$ methods, and asks whether to add a further, $n+1$st, method. Based on Equation (\ref{constweights}), the additional method will improve the combination if its individual performance is sufficiently good (i.e., $\mathbb{E}_\mathbb{Q}(S_{n+1})$ is sufficiently small), and the method is sufficiently different from the other methods (i.e., $\sum_{i=1}^n \mathbb{E}_\mathbb{Q}(d_{i,n+1})$ is sufficiently large). While these considerations may provide some guidance, principled techniques for selecting the members of the linear pool in practice remain an important challenge.

\subsection*{Linear versus Quantile Combination of Forecast Distributions} 

As discussed in Section \ref{sec:qp}, quantile combination produces a forecast distribution that is typically less dispersed than the linear pool. Under certain conditions, this is a desirable feature (see \citealt[Proposition 5]{LichtendahlEtAl2013} as well as Proposition \ref{lpvsqp} above). Conversely, if the individual forecasts are overconfident, the linear pool's increased dispersion is desirable. The empirical results on the relative performance of linear and quantile based combinations are mixed, which motivates \cite{Taylor2025} to propose a variant that encompasses both options. 

As discussed throughout this paper, several key properties of the linear pool apply to discrete and continuous, univariate and multivariate settings alike. The linear pool's versatility is appealing, and differs from quantile-based combinations which apply mostly to the univariate continuous case. 

\subsection*{Measures of Forecast Disagreement}
 
Divergences of proper scoring rules are useful for measuring differences between distributions \citep{ThorarinsdottirEtAl2013}. Here we argue to use kernel scores (i.e. a subclass of proper scoring rules) whose divergence functions are necessarily symmetric. Appendix \ref{sec:specific} provides details on the divergence functions for various kernel scores, covering both well-known measures (such as cross-sectional forecast variance for point forecasts and Cram\'er divergence for distributions) and less common ones (such as energy score divergence). Combined with Proposition \ref{prop0}, these formulas yield broadly applicable measures of forecast disagreement as studied in economics \citep[e.g.][]{ClarkMertens2024}.

\subsection*{CRPS} 

The CRPS is an increasingly popular scoring rule in statistics and machine learning. Our finding that various results on squared error generalize to the CRPS is potentially surprising: The CRPS is often viewed as a generalization of \textit{absolute} error, to which it reduces when the forecast distribution is of Dirac type, with point mass at a single value \citep[Section 4.2]{GneitingRaftery2007}. However, as discussed in Section \ref{sec:kernel}, absolute error is not a kernel score except in very special cases such as the Dirac case. Our results thus challenge the common intuition of the CRPS as being similar to absolute error, and instead highlights similarities between the CRPS and \textit{squared} error.

\subsection*{Computation} 

The methods discussed in this paper require to compute entropy and divergence functions associated with kernel scores (see Table \ref{tab:scores2} in the appendix). These computations are relatively straightforward for discrete samples that are popular in practice. Our replication material includes relevant software for \verb|R|, covering the CPRS and energy score in particular. The CRPS entropy and divergence for parametric distributions is more elusive, except for the Gaussian case which we cover in Appendix \ref{crps_gaussian}. \cite{AllenEtAl2024} show that estimating the weights under a kernel score can be cast as a convex quadratic programming problem. While our optimization problem is not convex quadratic in its original form considered in Section \ref{sec:implications_weights}, Appendix \ref{emp_weight} shows how to translate it into such a form. All of our computations use this latter form, together 
with the function \verb|ipop| of the \verb|R| package \verb|kernlab| due to \citealt{KaratzoglouEtAl2004}. 

\bibliographystyle{apalike}

\begin{thebibliography}{}

\bibitem[Abbas, 2009]{Abbas2009}
Abbas, A.~E. (2009).
\newblock A {K}ullback-{L}eibler view of linear and log-linear pools.
\newblock {\em Decision Analysis}, 6:25--37.

\bibitem[Abbate and Marcellino, 2018]{AbbateMarcellino2018}
Abbate, A. and Marcellino, M. (2018).
\newblock Point, interval and density forecasts of exchange rates with time
  varying parameter models.
\newblock {\em Journal of the Royal Statistical Society Series A: Statistics in
  Society}, 181:155--179.

\bibitem[Allen et~al., 2023]{AllenEtAl2023}
Allen, S., Ginsbourger, D., and Ziegel, J. (2023).
\newblock Evaluating forecasts for high-impact events using transformed kernel
  scores.
\newblock {\em SIAM/ASA Journal on Uncertainty Quantification}, 11:906--940.

\bibitem[Allen et~al., 2025]{AllenEtAl2024}
Allen, S., Ginsbourger, D., and Ziegel, J. (2025).
\newblock Efficient pooling of predictions via kernel embeddings.
\newblock {\em Transactions on Machine Learning Research}.

\bibitem[Banternghansa and McCracken, 2009]{banternghansa2009forecast}
Banternghansa, C. and McCracken, M.~W. (2009).
\newblock Forecast disagreement among {FOMC} members.
\newblock Federal Reserve Bank of St. Louis, working paper 2009-059A.

\bibitem[Bates and Granger, 1969]{Bates-Granger-69}
Bates, J.~M. and Granger, C. W.~J. (1969).
\newblock The combination of forecasts.
\newblock {\em Journal of the Operational Research Society}, 20:451--468.

\bibitem[Bellemare et~al., 2017]{BellemareEtAl2017}
Bellemare, M.~G., Danihelka, I., Dabney, W., Mohamed, S., Lakshminarayanan, B.,
  Hoyer, S., and Munos, R. (2017).
\newblock The {C}r\'amer distance as a solution to biased {W}asserstein
  gradients.
\newblock Preprint, arXiv:1705.10743.

\bibitem[Bentzien and Friederichs, 2014]{BentzienFriederichs2014}
Bentzien, S. and Friederichs, P. (2014).
\newblock Decomposition and graphical portrayal of the quantile score.
\newblock {\em Quarterly Journal of the Royal Meteorological Society},
  140:1924--1934.

\bibitem[Berrisch and Ziel, 2023]{BerrischZiel2023}
Berrisch, J. and Ziel, F. (2023).
\newblock {CRPS} learning.
\newblock {\em Journal of Econometrics}, 237:105221.

\bibitem[Brier, 1950]{Brier1950}
Brier, G.~W. (1950).
\newblock Verification of forecasts expressed in terms of probability.
\newblock {\em Monthly Weather Review}, 78:1--3.

\bibitem[Claeskens et~al., 2016]{ClaeskensEtAl2016}
Claeskens, G., Magnus, J.~R., Vasnev, A.~L., and Wang, W. (2016).
\newblock The forecast combination puzzle: A simple theoretical explanation.
\newblock {\em International Journal of Forecasting}, 32:754--762.

\bibitem[Clark and Mertens, 2024]{ClarkMertens2024}
Clark, T.~E. and Mertens, E. (2024).
\newblock Survey expectations and forecast uncertainty.
\newblock In Clements, M.~P. and Galv{\~a}o, A.~B., editors, {\em Handbook of
  Research Methods and Applications in Macroeconomic Forecasting}, pages
  305--333. Edward Elgar Publishing.

\bibitem[Clements et~al., 2023]{Clements2023}
Clements, M.~P., Rich, R.~W., and Tracy, J.~S. (2023).
\newblock Surveys of professionals.
\newblock In Bachmann, R., van~der Klaauw, W., and Topa, G., editors, {\em
  Handbook of Economic Expectations}, pages 71--106. Elsevier.

\bibitem[Coibion and Gorodnichenko, 2012]{CoibionEtAl2012}
Coibion, O. and Gorodnichenko, Y. (2012).
\newblock What can survey forecasts tell us about information rigidities?
\newblock {\em Journal of Political Economy}, 120:116--159.

\bibitem[Cramer et~al., 2022]{CramerEtAl2022}
Cramer, E.~Y., Ray, E.~L., Lopez, V.~K., Bracher, J., Brennen, A.,
  Castro~Rivadeneira, A.~J., Gerding, A., Gneiting, T., House, K.~H., Huang,
  Y., et~al. (2022).
\newblock Evaluation of individual and ensemble probabilistic forecasts of
  {COVID-19} mortality in the {U}nited {S}tates.
\newblock {\em Proceedings of the National Academy of Sciences},
  119:e2113561119.

\bibitem[Cumings-Menon et~al., 2021]{CumingsEtAl2021}
Cumings-Menon, R., Shin, M., and Sill, D.~K. (2021).
\newblock Measuring disagreement in probabilistic and density forecasts.
\newblock Federal Reserve Bank of Philadelphia, working paper 21-03.

\bibitem[Dawid, 2007]{Dawid2007}
Dawid, A.~P. (2007).
\newblock The geometry of proper scoring rules.
\newblock {\em Annals of the Institute of Statistical Mathematics}, 59:77--93.

\bibitem[Del~Negro and Primiceri, 2015]{delNegro2015}
Del~Negro, M. and Primiceri, G.~E. (2015).
\newblock Time varying structural vector autoregressions and monetary policy: A
  corrigendum.
\newblock {\em The Review of Economic Studies}, 82:1342--1345.

\bibitem[Diebold and Mariano, 1995]{DieboldMariano1995}
Diebold, F.~X. and Mariano, R.~S. (1995).
\newblock Comparing predictive accuracy.
\newblock {\em Journal of Business \& Economic Statistics}, 13:253--263.

\bibitem[Diebold et~al., 2025]{DieboldEtAl2025}
Diebold, F.~X., Mora, A., and Shin, M. (2025).
\newblock On the wisdom of crowds (of economists).
\newblock Preprint, arXiv:2503.09287.

\bibitem[Elliott et~al., 2013]{ElliottEtAl2013}
Elliott, G., Gargano, A., and Timmermann, A. (2013).
\newblock Complete subset regressions.
\newblock {\em Journal of Econometrics}, 177:357--373.

\bibitem[Elliott and Liao, 2025]{Elliott2025}
Elliott, G. and Liao, J. (2025).
\newblock Combining forecasts-on why averaging beats optimal linear weights.
\newblock Preprint, available at
  \url{https://econweb.ucsd.edu/~gelliott/Combination.pdf} (version May 5,
  2025).

\bibitem[Engle, 1983]{Engle1983}
Engle, R.~F. (1983).
\newblock Estimates of the variance of {US} inflation based upon the {ARCH}
  model.
\newblock {\em Journal of Money, Credit and Banking}, 15:286--301.

\bibitem[Epstein, 1969]{Epstein1969}
Epstein, E.~S. (1969).
\newblock A scoring system for probability forecasts of ranked categories.
\newblock {\em Journal of Applied Meteorology}, 8:985--987.

\bibitem[Fakoor et~al., 2023]{FakoorEtAl2023}
Fakoor, R., Kim, T., Mueller, J., Smola, A.~J., and Tibshirani, R.~J. (2023).
\newblock Flexible model aggregation for quantile regression.
\newblock {\em Journal of Machine Learning Research}, 24:1--45.

\bibitem[{Federal Reserve Bank of Philadelphia}, 2025a]{PhiladelphiaFed2024b}
{Federal Reserve Bank of Philadelphia} (2025a).
\newblock Real-time data set for macroeconomists.
\newblock Data set, available at
  \url{https://www.philadelphiafed.org/surveys-and-data/real-time-data-research/real-time-data-set-for-macroeconomists}
  (last accessed: January 3, 2025).

\bibitem[{Federal Reserve Bank of Philadelphia}, 2025b]{PhiladelphiaFed2024a}
{Federal Reserve Bank of Philadelphia} (2025b).
\newblock Survey of professional forecasters.
\newblock Data set, available at
  \url{https://www.philadelphiafed.org/surveys-and-data/real-time-data-research/survey-of-professional-forecasters}
  (last accessed: January 3, 2025).

\bibitem[{Federal Reserve Bank of St. Louis}, 2026]{fred2026}
{Federal Reserve Bank of St. Louis} (2026).
\newblock {FRED}.
\newblock Data base, available at \url{https://fred.stlouisfed.org}.

\bibitem[Genest and Zidek, 1986]{genest1986combining}
Genest, C. and Zidek, J.~V. (1986).
\newblock Combining probability distributions: A critique and an annotated
  bibliography.
\newblock {\em Statistical Science}, 1:114--135.

\bibitem[Geweke and Amisano, 2011]{GewekeAmisano2011}
Geweke, J. and Amisano, G. (2011).
\newblock Optimal prediction pools.
\newblock {\em Journal of Econometrics}, 164:130--141.

\bibitem[Giacomini and Rossi, 2010]{GiacominiRossi2010}
Giacomini, R. and Rossi, B. (2010).
\newblock Forecast comparisons in unstable environments.
\newblock {\em Journal of Applied Econometrics}, 25:595--620.

\bibitem[Gneiting, 2012]{Gneiting2012}
Gneiting, T. (2012).
\newblock On the cover-hart inequality: What's a sample of size one worth?
\newblock {\em Stat}, 1:12--17.

\bibitem[Gneiting and Raftery, 2007]{GneitingRaftery2007}
Gneiting, T. and Raftery, A.~E. (2007).
\newblock Strictly proper scoring rules, prediction, and estimation.
\newblock {\em Journal of the American Statistical Association}, 102:359--378.

\bibitem[Gneiting and Ranjan, 2011]{GneitingRanjan2011}
Gneiting, T. and Ranjan, R. (2011).
\newblock Comparing density forecasts using threshold-and quantile-weighted
  scoring rules.
\newblock {\em Journal of Business \& Economic Statistics}, 29:411--422.

\bibitem[Gneiting and Ranjan, 2013]{GneitingRanjan2013}
Gneiting, T. and Ranjan, R. (2013).
\newblock Combining predictive distributions.
\newblock {\em Electronic Journal of Statistics}, 7:1747--1782.

\bibitem[Gneiting and Resin, 2023]{GneitingResin2023}
Gneiting, T. and Resin, J. (2023).
\newblock Regression diagnostics meets forecast evaluation: Conditional
  calibration, reliability diagrams, and coefficient of determination.
\newblock {\em Electronic Journal of Statistics}, 17:3226--3286.

\bibitem[Grushka-Cockayne et~al., 2017]{GrushkaEtAl2017}
Grushka-Cockayne, Y., Lichtendahl~Jr, K.~C., Jose, V. R.~R., and Winkler, R.~L.
  (2017).
\newblock Quantile evaluation, sensitivity to bracketing, and sharing business
  payoffs.
\newblock {\em Operations Research}, 65:712--728.

\bibitem[Hall and Mitchell, 2007]{HallMitchell2007}
Hall, S.~G. and Mitchell, J. (2007).
\newblock Combining density forecasts.
\newblock {\em International Journal of Forecasting}, 23:1--13.

\bibitem[Hansen and Lunde, 2005]{HansenLunde2005}
Hansen, P.~R. and Lunde, A. (2005).
\newblock A forecast comparison of volatility models: does anything beat a
  {GARCH}(1, 1)?
\newblock {\em Journal of Applied Econometrics}, 20:873--889.

\bibitem[Held et~al., 2010]{HeldEtAl2010}
Held, L., Rufibach, K., and Balabdaoui, F. (2010).
\newblock A score regression approach to assess calibration of continuous
  probabilistic predictions.
\newblock {\em Biometrics}, 66:1295--1305.

\bibitem[Holzmann and Eulert, 2014]{HolzmannEulert2014}
Holzmann, H. and Eulert, M. (2014).
\newblock The role of the information set for forecasting—with applications
  to risk management.
\newblock {\em The Annals of Applied Statistics}, pages 595--621.

\bibitem[Hora, 2004]{Hora2004}
Hora, S.~C. (2004).
\newblock Probability judgments for continuous quantities: Linear combinations
  and calibration.
\newblock {\em Management Science}, 50:597--604.

\bibitem[Huber and Z{\"o}rner, 2019]{Huber2019}
Huber, F. and Z{\"o}rner, T.~O. (2019).
\newblock Threshold cointegration in international exchange rates: A {B}ayesian
  approach.
\newblock {\em International Journal of Forecasting}, 35:458--473.

\bibitem[Karatzoglou et~al., 2004]{KaratzoglouEtAl2004}
Karatzoglou, A., Smola, A., Hornik, K., and Zeileis, A. (2004).
\newblock kernlab-an {S4} package for kernel methods in {R}.
\newblock {\em Journal of Statistical Software}, 11:1--20.

\bibitem[Kn{\"u}ppel et~al., 2022]{KnueppelEtAl2022}
Kn{\"u}ppel, M., Kr{\"u}ger, F., and Pohle, M.-O. (2022).
\newblock Score-based calibration testing for multivariate forecast
  distributions.
\newblock Preprint, arXiv:2211.16362.

\bibitem[Knüppel and Krüger, 2022]{KnueppelKrueger2022}
Knüppel, M. and Krüger, F. (2022).
\newblock Forecast uncertainty, disagreement, and the linear pool.
\newblock {\em Journal of Applied Econometrics}, 37:23--41.

\bibitem[Koenker and Bassett~Jr, 1978]{KoenkerBassett1978}
Koenker, R. and Bassett~Jr, G. (1978).
\newblock Regression quantiles.
\newblock {\em Econometrica}, 46:33--50.

\bibitem[Kr\"uger, 2015]{Krueger2015}
Kr\"uger, F. (2015).
\newblock {\em bvarsv: {B}ayesian Analysis of a Vector Autoregressive Model
  with Stochastic Volatility and Time-Varying Parameters}.
\newblock {R} package version 1.1.

\bibitem[Kr{\"u}ger et~al., 2017]{KruegerEtAl2017}
Kr{\"u}ger, F., Clark, T.~E., and Ravazzolo, F. (2017).
\newblock Using entropic tilting to combine {BVAR} forecasts with external
  nowcasts.
\newblock {\em Journal of Business \& Economic Statistics}, 35:470--485.

\bibitem[Kr{\"u}ger and Pavlova, 2024]{Krueger2024}
Kr{\"u}ger, F. and Pavlova, L. (2024).
\newblock Quantifying subjective uncertainty in survey expectations.
\newblock {\em International Journal of Forecasting}, 40:796--810.

\bibitem[Kr{\"u}ger and Ziegel, 2021]{KruegerZiegel2021}
Kr{\"u}ger, F. and Ziegel, J.~F. (2021).
\newblock Generic conditions for forecast dominance.
\newblock {\em Journal of Business \& Economic Statistics}, 39:972--983.

\bibitem[Krüger, 2013]{Krueger2013}
Krüger, F. (2013).
\newblock {\em Four Essays on Probabilistic Forecasting in Econometrics}.
\newblock PhD thesis, Universität Konstanz.

\bibitem[Lahiri and Sheng, 2010]{lahiri2010}
Lahiri, K. and Sheng, X. (2010).
\newblock Measuring forecast uncertainty by disagreement: The missing link.
\newblock {\em Journal of Applied Econometrics}, 25:514--538.

\bibitem[Lichtendahl~Jr et~al., 2013]{LichtendahlEtAl2013}
Lichtendahl~Jr, K.~C., Grushka-Cockayne, Y., and Winkler, R.~L. (2013).
\newblock Is it better to average probabilities or quantiles?
\newblock {\em Management Science}, 59:1594--1611.

\bibitem[Matheson and Winkler, 1976]{MathesonWinkler1976}
Matheson, J.~E. and Winkler, R.~L. (1976).
\newblock Scoring rules for continuous probability distributions.
\newblock {\em Management Science}, 22:1087--1096.

\bibitem[Mitchell et~al., 2026]{MitchellEtAl2024}
Mitchell, J., Shiroff, T., and Braitsch, H. (2026).
\newblock Practice makes perfect: Learning effects with household point and
  density forecasts of inflation.
\newblock {\em International Journal of Forecasting}, 42:315--329.

\bibitem[Mosler and Mozharovskyi, 2022]{Mosler2022}
Mosler, K. and Mozharovskyi, P. (2022).
\newblock Choosing among notions of multivariate depth statistics.
\newblock {\em Statistical Science}, 37:348--368.

\bibitem[Murphy and Winkler, 1987]{MurphyWinkler1987}
Murphy, A.~H. and Winkler, R.~L. (1987).
\newblock A general framework for forecast verification.
\newblock {\em Monthly Weather Review}, 115:1330--1338.

\bibitem[Neyman and Roughgarden, 2023]{NeymanRoughgarden2023}
Neyman, E. and Roughgarden, T. (2023).
\newblock From proper scoring rules to max-min optimal forecast aggregation.
\newblock {\em Operations Research}, 71:2175--2195.

\bibitem[Pesaran and Pick, 2011]{Pesaran2011}
Pesaran, M.~H. and Pick, A. (2011).
\newblock Forecast combination across estimation windows.
\newblock {\em Journal of Business \& Economic Statistics}, 29:307--318.

\bibitem[Pettigrew, 2019]{Pettigrew2019}
Pettigrew, R. (2019).
\newblock Aggregating incoherent agents who disagree.
\newblock {\em Synthese}, 196:2737--2776.

\bibitem[Primiceri, 2005]{Primiceri2005}
Primiceri, G.~E. (2005).
\newblock Time varying structural vector autoregressions and monetary policy.
\newblock {\em The Review of Economic Studies}, 72:821--852.

\bibitem[Ranjan and Gneiting, 2010]{RanjanGneiting2010}
Ranjan, R. and Gneiting, T. (2010).
\newblock Combining probability forecasts.
\newblock {\em Journal of the Royal Statistical Society Series B: Statistical
  Methodology}, 72:71--91.

\bibitem[Resin et~al., 2026]{ResinEtAl2024}
Resin, J., Wolffram, D., Bracher, J., and Dimitriadis, T. (2026).
\newblock Shift-dispersion decompositions of {W}asserstein and {C}ram\'er
  distances.
\newblock {\em Statistical Science}.

\bibitem[Rossi, 2013]{Rossi2013}
Rossi, B. (2013).
\newblock Exchange rate predictability.
\newblock {\em Journal of Economic Literature}, 51:1063--1119.

\bibitem[Schefzik et~al., 2013]{SchefzikEtAl2013}
Schefzik, R., Thorarinsdottir, T.~L., and Gneiting, T. (2013).
\newblock Uncertainty quantification in complex simulation models using
  ensemble copula coupling.
\newblock {\em Statistical Science}, 28:616--640.

\bibitem[Sejdinovic et~al., 2013]{SejdinovicEtAl2013}
Sejdinovic, D., Sriperumbudur, B., Gretton, A., and Fukumizu, K. (2013).
\newblock Equivalence of distance-based and {RKHS}-based statistics in
  hypothesis testing.
\newblock {\em The Annals of Statistics}, 41:2263--2291.

\bibitem[Shoja and Soofi, 2017]{ShojaSoofi2017}
Shoja, M. and Soofi, E.~S. (2017).
\newblock Uncertainty, information, and disagreement of economic forecasters.
\newblock {\em Econometric Reviews}, 36:796--817.

\bibitem[Stone, 1961]{Stone1961}
Stone, M. (1961).
\newblock The opinion pool.
\newblock {\em The Annals of Mathematical Statistics}, 32:1339--1342.

\bibitem[Sz{\'e}kely and Rizzo, 2017]{SzekelyRizzo2017}
Sz{\'e}kely, G.~J. and Rizzo, M.~L. (2017).
\newblock The energy of data.
\newblock {\em Annual Review of Statistics and Its Application}, 4:447--479.

\bibitem[Taylor, 2026]{Taylor2025b}
Taylor, J.~W. (2026).
\newblock Probabilistic forecast aggregation with statistical depth.
\newblock {\em European Journal of Operational Research}, 328:460--476.

\bibitem[Taylor and Meng, 2025]{Taylor2025}
Taylor, J.~W. and Meng, X. (2025).
\newblock Angular combining of forecasts of probability distributions.
\newblock {\em Management Science}, 72:2111--2133.

\bibitem[Thorarinsdottir et~al., 2013]{ThorarinsdottirEtAl2013}
Thorarinsdottir, T.~L., Gneiting, T., and Gissibl, N. (2013).
\newblock Using proper divergence functions to evaluate climate models.
\newblock {\em SIAM/ASA Journal on Uncertainty Quantification}, 1:522--534.

\bibitem[Tibshirani, 2023]{Tibshirani2023}
Tibshirani, R. (2023).
\newblock Forecast scoring and calibration.
\newblock Course notes, University of California at Berkeley, available at
  \url{https://www.stat.berkeley.edu/~ryantibs/statlearn-s23/lectures/calibration.pdf}.

\bibitem[Timmermann, 2006]{Timmermann2006}
Timmermann, A. (2006).
\newblock Forecast combinations.
\newblock In Elliott, G., Granger, C.~W., and Timmermann, A., editors, {\em
  Handbook of Economic Forecasting}, volume~1, chapter~4, pages 135--196.
  Elsevier.

\bibitem[Tsyplakov, 2013]{Tsyplakov2013}
Tsyplakov, A. (2013).
\newblock Evaluation of probabilistic forecasts: proper scoring rules and
  moments.
\newblock Preprint, available at \url{https://dx.doi.org/10.2139/ssrn.2236605}.

\bibitem[Waghmare and Ziegel, 2026]{WaghmareZiegel2025}
Waghmare, K. and Ziegel, J. (2026).
\newblock Proper scoring rules for estimation and forecast evaluation.
\newblock {\em Annual Review of Statistics and Its Application}, 13:271--296.

\bibitem[Wallis, 2005]{wallis2005}
Wallis, K.~F. (2005).
\newblock Combining density and interval forecasts: A modest proposal.
\newblock {\em Oxford Bulletin of Economics and Statistics}, 67:983--994.

\bibitem[Wang et~al., 2023]{WangEtAl2023}
Wang, X., Hyndman, R.~J., Li, F., and Kang, Y. (2023).
\newblock Forecast combinations: An over 50-year review.
\newblock {\em International Journal of Forecasting}, 39:1518--1547.

\bibitem[Zarnowitz and Lambros, 1987]{ZarnowitzLambros1987}
Zarnowitz, V. and Lambros, L.~A. (1987).
\newblock Consensus and uncertainty in economic prediction.
\newblock {\em Journal of Political Economy}, 95:591--621.

\bibitem[Zeileis, 2004]{Zeileis2004}
Zeileis, A. (2004).
\newblock Econometric computing with {HC} and {HAC} covariance matrix
  estimators.
\newblock {\em Journal of Statistical Software}, 11:1--17.

\bibitem[Zeileis et~al., 2020]{ZeileisEtAl2020}
Zeileis, A., K{\"o}ll, S., and Graham, N. (2020).
\newblock Various versatile variances: an object-oriented implementation of
  clustered covariances in {R}.
\newblock {\em Journal of Statistical Software}, 95:1--36.

\end{thebibliography}

\newpage

\begin{appendix}
	
\section*{Appendix}
	
\setcounter{page}{1}
	
\section{Divergence and Entropy for Specific Scoring Rules}\label{sec:specific}

This section provides specifics for various scoring rules $S_L$ of applied interest (see Table \ref{tab:scores}). While Table \ref{tab:scores2} presents expressions for the divergence and entropy of each scoring rule, Sections \ref{sec:se} to \ref{sec:energy_score} provide further context.

\begin{table}[h]
	\centering
	\renewcommand{\arraystretch}{1.8}
	\begin{tabular}{lcc}
		Scoring rule & Divergence & Entropy \\ \toprule
		SE & $(\mu_a-\mu_b)^2$ & $\sigma^2$ \\
		Mult. SE& $(\mu_a-\mu_b)'A(\mu_a-\mu_b)$& $\text{trace}(A\Sigma)$\\
		CRPS & $\int_{-\infty}^\infty (F_a(z)-F_b(z))^2~dz$ & $\int_{-\infty}^\infty F(z)(1-F(z))~dz$\\
		Energy Score & $\mathbb{E}_{F_a,F_b}\Big[||\widetilde{X}-X||\Big] - \frac{1}{2} \mathbb{E}_{F_a}\Big[||\widetilde{X}-X||\Big] - \frac{1}{2}\mathbb{E}_{F_b}\Big[||\widetilde{X}-X||\Big]$ &
 $\frac{1}{2}\mathbb{E}_{F}\Big[||\widetilde{X}-X||\Big]$\\
		Brier Score & $\frac{1}{2} \sum_{l=1}^k (p_{l,a}-p_{l,b})^2$ & $\frac{1}{2} \sum_{l=1}^k p_l(1-p_l)$  \\
		RPS & $\sum_{l=1}^k (P_{l,a} - P_{l,b})^2$ & $\sum_{l=1}^k P_l(1-P_l)$ \\ \bottomrule
	\end{tabular}
	\caption{Divergence and entropy terms for several kernel scoring rules. Indices $a$ and $b$ denote two forecast distributions. See Sections \ref{sec:se} to \ref{sec:rps} for context and details on notation. \label{tab:scores2}}
\end{table}

\subsection{Squared Error}\label{sec:se}

As noted by \citet[p.~14]{Gneiting2012}, squared error loss corresponds to the kernel function $L(z, \widetilde{z}) = (\widetilde{z}-z)^2$ and real-valued univariate outcomes, i.e. $\Omega = \mathbb{R}$. To see this, observe that $\mathbb{E}_F[L(X,y)] = \mathbb{E}_F\big[(y-X)^2\big] = y^2 - 2y\mu + \mu^2 + \sigma^2,$ where $\mu= \mathbb{E}_F[X]$ and $\sigma^2 = V_F[X]$. Furthermore, $\mathbb{E}_F[L(X,\widetilde{X})] = 2 \sigma^2$. From Equation (\ref{score}), we thus obtain squared error loss $\text{SE}(F, y) = (y-\mu)^2.$

Thus for squared error, entropy coincides with variance. The linear pool's variance, and its decomposition into `average variance' and `disagreement between point forecasts', has been discussed in detail in the literature on macroeconomic survey forecasts, see e.g. \cite{ZarnowitzLambros1987}, \cite{wallis2005} and \cite{lahiri2010}.

\subsection{Multivariate Squared Error}\label{sec:se_mult}

Let $\Omega = \mathbb{R}^k$, where $k$ is a finite integer, and consider the kernel function $$L_A(z, \widetilde{z}) = (\widetilde{z}-z)'A(\widetilde{z}-z),$$ where $A$ is a symmetric positive definite matrix. If $A = I_k$, this is the generalized version of the Energy Score, with $\beta = 2$ in the notation of \citet[Section 5.1]{GneitingRaftery2007}. The kernel function $L_A$ remains negative definite for any positive definite matrix $A$. (To see this, consider the Cholesky decomposition of $A = G G'$. Then $L_A(z, \widetilde z) = (\widetilde z'G - z'G)(G'\widetilde z - G' z) = L(u, \widetilde u)$ for $u = G'z, \widetilde u = G'\widetilde z$. The definition of negative definite kernels and the fact that $L(z, \widetilde z) = ||\widetilde z -  z||^2 = \sqrt{\sum_{l=1}^k (\widetilde z_l-z_l)^2}$ is a negative definite kernel on $\mathbb{R}^k$ \citep[see][Table 1]{Gneiting2012} then imply that $L_A(z, z)$ is a negative definite kernel on $\mathbb{R}^k$ as well.)
The kernel function yields the scoring rule 
$$S_{L_A}(F, y) = \left(y - \mu\right)'A\left(y - \mu\right),$$
which evaluates the $k$-variate mean vector $\mu = \mathbb{E}_{F}[X]$ implied by $F$. This scoring rule corresponds to the negative log likelihood of a $k$-variate Gaussian random variable with known covariance matrix $A^{-1}$. It is potentially useful to aggregate forecasting performance across several variables (elements of $X$), with the matrix $A$ accounting for scale differences across variables, or correlation between them. The divergence function for this scoring rule is given by 
$$d(F_a, F_b) = (\mu_a - \mu_b)'A(\mu_a - \mu_b),$$
where $F_a$ and $F_b$ are two distributions with respective mean vectors $\mu_a$ and $\mu_b$. When $A$ is set to the inverse of an appropriate empirical covariance matrix, $\sqrt{d(F_a, F_b)}$ is the Mahalonobis distance between $F_a$ and $F_b$. The latter has been used by studies such as \cite{banternghansa2009forecast} and \cite{Clements2023} to measure multivariate forecast disagreement. Interestingly, though, our expression for $D$ suggests to use the square of the Mahalonobis distance instead.

\subsection{CRPS}\label{sec:crps}

The CRPS \citep{MathesonWinkler1976} corresponds to the kernel function $L(z, \widetilde{z}) = |\widetilde{z}-z|$ and $\Omega = \mathbb{R}$ \citep[Section 5.1]{GneitingRaftery2007}. From (\ref{score}), the CRPS is given by
$$\text{CRPS}(F, y) = \mathbb{E}_F\Big[|y-X|\Big]- \frac{1}{2} \mathbb{E}_F\Big[|\widetilde{X}-X|\Big],$$
where $X$ and $\widetilde{X}$ are two independent draws from $F$. This formula is often called the kernel representation of the CRPS. An alternative, equivalent representation of the CPRS is 
$$\text{CRPS}(F, y) = \int_{-\infty}^\infty (\mathbf{1}(z \ge y) - F(z))^2~dz,$$
where $\mathbf{1}(A)$ is the indicator function of the event $A$. The entropy function of the CRPS is given by 
$$\mathbb{E}_{F}[\text{CRPS}(F, X)] = \frac{1}{2}\mathbb{E}_F\Big[|\widetilde{X}-X|\Big] = \int_{-\infty}^\infty F(z)(1-F(z))~dz.$$

The divergence function associated with the CRPS, $$d(F_a, F_b) = \int_{-\infty}^\infty (F_a(z)-F_b(z))^2~dz,$$
is the Cram\'er distance between $F_a$ and $F_b$; see \cite{ThorarinsdottirEtAl2013}, \cite{BellemareEtAl2017} and \cite{ResinEtAl2024} for properties and applications. 

\subsection{Energy Score}\label{sec:energy_score}

The Energy Score is a multivariate generalization of the CRPS. It is a kernel score with $L(z,\widetilde{z}) = ||\widetilde{z}-z|| = \sqrt{\sum_{l=1}^k (\widetilde z_l-z_l)^2}$ and $\Omega = \mathbb{R}^k$ \citep[Section 5.1]{GneitingRaftery2007}. \citeauthor{GneitingRaftery2007} consider a more general formulation of the Energy Score. The variant considered here obtains when setting $\beta = 1$ (in their notation). From (\ref{score}), the Energy Score is given by 
$$\text{ES}(F, y) = \mathbb{E}_{F}\Big[||y - X||\Big] - \frac{1}{2}\mathbb{E}_F\Big[||\widetilde X - X||\Big];$$
as before, $X$ and $\widetilde{X}$ are two independent draws from $F$.
\cite{KnueppelEtAl2022} propose using the Energy Score's entropy function for testing the calibration of multivariate forecast distributions. 
For the Energy Score, the disagreement term $D$ at (\ref{disagreement_lp}) is given by 
$$\frac{1}{2}\mathbb{E}_{F_w}\Big[||\widetilde{X} - X||\Big] - \frac{1}{2}\sum_{i=1}^n w_i \mathbb{E}_{F_i}\Big[||\widetilde{X}-X||\Big].$$
While we are not aware of existing applications of this disagreement measure, it seems useful for comparing multivariate forecast distributions, as considered by \cite{CumingsEtAl2021} in the context of vector autoregressions for macroeconomic variables. We present empirical evidence on this disagreement measure in Section \ref{sec:illustration}.

\subsection{Brier Score}\label{sec:brier_score}

The Brier score is a scoring rule for probabilities of unordered categorical outcomes. While it is most popular in the binary case originally considered by \cite{Brier1950}, it is readily applicable to an outcome variable $Y$ that  takes $k$ distinct values. The Brier score obtains when setting $\Omega = \{1, 2, \ldots, k\}$ and $L(z, \widetilde{z}) = \mathbf{1}(\widetilde{z} \neq z)$ \citep{Gneiting2012}. Importantly, the outcomes in $\Omega$ are interpreted as interchangeable labels, rather than integers. In Table \ref{tab:scores2}, we identify a forecast distribution $F$ of a categorical outcome with a $k \times 1$ vector $\underline{p} = \begin{pmatrix} p_{1}, &\ldots,& p_{k}\end{pmatrix}'$, such that the elements of $\underline{p}$ are nonnegative and sum to one. 

\subsection{Ranked Probability Score}\label{sec:rps}

We next consider the case of ordered categorical outcomes. Examples include bond ratings in finance and binned numerical data (see \citealt{Krueger2024}). In this setup, the outcome space $\Omega = \{1, 2, \ldots, k\}$ remains the same as for the Brier score, but the outcomes in $\Omega$ are interpreted as ordinal. This means that any two outcomes can be ranked, but it is not possible to quantify their difference. The Ranked Probability Score \citep[RPS;][]{Epstein1969} is tailored to this setup. Its kernel function $L(z, \widetilde z) = |\widetilde z - z|$ is the same as for the CRPS, but it is used in conjunction with the specific outcome space $\Omega$ just described. In Table \ref{tab:scores2}, we denote by $P_{l} = \sum_{r=1}^l p_{r}$ the cumulative probability of the first $l$ elements of $\underline{p}$. The RPS' use of cumulative probabilities reflects the fact that it attaches an ordinal interpretation to the categories. This feature is distinct from the Brier Score, which views the categories' labels as interchangeable. \cite{Krueger2024} recommend using the RPS' entropy function, which they call ERPS, as an uncertainty measure for binned macroeconomic data. In an insightful discussion of their empirical results on inflation expectations, \citet[Section 3.2]{MitchellEtAl2024} recently conjectured that a disagreement-variance type decomposition of the linear pool's ERPS exists. Our results confirm this conjecture.

\section{Proofs}

\label{sec:proofs}

\subsection*{Proposition \ref{prop:d}}

\begin{eqnarray*}
	D_\text{gen}(h) -D_{\text{gen}}(p_w) &=& {p_w}' \underline{L}h - \frac{1}{2} h' \underline{L}h- \frac{1}{2}{p_w}' \underline{L} p_w \\
	&=& -\frac{1}{2}(h-p_w)' \underline{L}(h-p_w)\\
	&=& -\frac{1}{2} \sum_{j=1}^{n_\Omega}\sum_{l=1}^{n_\Omega} c_jc_l~L(x_{(j)}, x_{(l)}) \\
	&\ge & 0,
\end{eqnarray*}
where $c_j = h_j - p_{w,j}$ is the difference between the $j$th elements of $h$ and $p_w$, with $\sum_{j=1}^{n_\Omega} c_j = 1-1= 0$. The inequality then follows from the definition of a negative definite kernel function $L$ (see Section \ref{sec:kernel}).

\subsection*{Proposition \ref{prop0}}

The first statement of part (a) follows from Equation (\ref{divergence}). The second statement holds because $S_L(F, y) \ge 0$ (see Section \ref{sec:kernel}). As regards part (b), Equation (\ref{disagreement_lp}) follows from Equations (\ref{divergence}) and (\ref{entropy_pool}). Equation (\ref{disagreement_pairs_lp}) holds because 
\begin{eqnarray*}
	D &=& \sum_{i=1}^n w_i \left\{\mathbb{E}_{F_i, F_w}[L(X, \widetilde{X})] - \frac{1}{2}\mathbb{E}_{F_i}[L(X,\widetilde{X})]- \frac{1}{2}\mathbb{E}_{F_w}[L(X,\widetilde{X})]\right\}\nonumber\\ 
	&=& \frac{1}{2}\mathbb{E}_{F_w}[L(X,\widetilde{X})] - \frac{1}{2} \sum_{i=1}^n w_i~\mathbb{E}_{F_i}[L(X, \widetilde{X})] \\
	&=& \frac{1}{2}\sum_{i=1}^n\sum_{j=1}^n w_iw_j \underbrace{\left\{\mathbb{E}_{F_i, F_j}[L(X, \widetilde{X})] - \frac{1}{2}\mathbb{E}_{F_i}[L(X,\widetilde{X})]- \frac{1}{2}\mathbb{E}_{F_j}[L(X,\widetilde{X})]\right\}}_{=d(F_i, F_j)}\nonumber\\
	&=& \sum_{i=1}^{n-1}\sum_{j>i} w_iw_j~d(F_i, F_j)\nonumber.
\end{eqnarray*}	

\subsection*{Proposition \ref{co_ent}}

\begin{eqnarray*}
\text{ent}(F_w^q) &=& \int_{-\infty}^\infty S(F_w^q, y) ~dF_w^q(y) \\
& \le & \int_{-\infty}^\infty S(F_w, y) ~dF_w^q(y) \\
& \le & \int_{-\infty}^\infty S(F_w, y) ~dF_w(y)\\
& = & \text{ent}(F_w),
\end{eqnarray*}
where the first inequality holds because $S$ is proper, and the second inequality holds because $S(F_w, y)$ is convex with respect to $y$ and $F_w$ is greater than $F_w^q$ in convex order \citep[Lemma 1]{LichtendahlEtAl2013}.

\subsection*{Example for Section \ref{sec:qp}}

Here we provide an example to demonstrate that the entropy of a quantile based combination $F_w^q$ may be strictly larger than the entropy of the linear pool $F_w$. The example is based on the threshold weighted CRPS \citep[twCRPS;][]{GneitingRanjan2011} given by
$$\text{twCRPS}(F, y) = \int_{-\infty}^\infty (\mathbf{1}(z \ge y - F(z))^2 \nu(z) dz,$$
where $\nu: \mathbb{R} \rightarrow \mathbb{R}_+$ is a nonnegative weight function. The twCRPS is a proper (but not necessarily strictly proper) scoring rule relative to typical classes of distributions. \citet[Proposition 1]{AllenEtAl2023} show that the twCRPS is a kernel score. Its entropy function is given by 
$$\text{ent}(F) = \int_{-\infty}^\infty \nu(z) F(z)(1-F(z))~dz.$$
The latter formula shows that if there is a subset $A \subset \mathbb{R}$ so that $F_w^q(z)(1-F_w^q(z)) > F_w(z)(1-F_w(z))$ for all $z \in A$, then using the weight function $\nu(z) = \mathbf{1}(z \in A)$ will yield an example with $\text{ent}(F_w^q) > \text{ent}(F_w)$. One such example can be obtained for $n = 2, w = 0.5, F_1 = \mathcal{N}(-1, 0.01), F_2 = \mathcal{N}(1, 4)$. In this case, the equally weighted quantile combination $F_w^q = \mathcal{N}(0, 1.1025)$ is also Gaussian \citep[see e.g.][]{LichtendahlEtAl2013} whereas the equally weighted linear pool is bimodal. The above inequality is satisfied, for example, on the interval $A = [-0.1, 0.1]$. 

\subsection*{Proposition \ref{prop1}}

\begin{eqnarray*}
	\sum_{i=1}^n w_i S_L(F_i, y) &=& \sum_{i=1}^n w_i \left(\mathbb{E}_{F_i}[L(X, y)] - \frac{1}{2}\mathbb{E}_{F_i}[L(X, \widetilde{X})]\right) \\
	&=& \mathbb{E}_{F_w}[L(X, y)] - \frac{1}{2}\mathbb{E}_{F_w}[L(X, \widetilde{X})] + \\
	&&  \mathbb{E}_{F_w}[L(X, \widetilde{X})] - \frac{1}{2} \mathbb{E}_{F_w}[L(X, \widetilde{X})] - \frac{1}{2} \sum_{i=1}^n w_i \mathbb{E}_{F_i}[L(X, \widetilde{X})] \\
	&=& S_L(F_w, y) + \\
	&& \underbrace{\sum_{i=1}^n w_i\left(\mathbb{E}_{F_i, F_w}[L(X, \widetilde{X})] - \frac{1}{2} \mathbb{E}_{F_w}[L(X, \widetilde{X})] - \frac{1}{2} \mathbb{E}_{F_i}[L(X, \widetilde{X})]\right)}_{= D},
\end{eqnarray*}
where the second equality uses the definition of the linear pool and the third equality uses the definitions in Equations (\ref{score}) and (\ref{disagreement_lp}). 

\subsection*{Proposition \ref{prop:equidist}}

We can write $\mathbb{E}_\mathbb{Q}(D) = \frac{1}{2}~w'\mathbf{ED}w$, where $\mathbf{ED}$ is an $n \times n$ matrix with $[i,j]$ entry given by $\mathbb{E}_\mathbb{Q}[d_{ij}].$ Under Assumption A, Proposition \ref{prop1} implies that the optimal combination weights maximize $w'\mathbf{ED}w$, as claimed in the first statement of Proposition \ref{prop:equidist}. Assumption B further implies that $\mathbf{ED} = (\tau\tau' -I_n)~d,$ where $\tau$ is an $n \times 1$ vector of ones, $I_n$ is the $n$-dimensional identity matrix and $d \in \mathbb{R}_+$ is a nonnegative scalar. Thus $w'\mathbf{ED}w = w'\tau\tau'w~d - w'wd = (1-w'w)~d.$ Hence the optimal combination weights minimize $w'w$, subject to the constraints of being nonnegative and summing to one. The latter problem is solved by equal weights, i.e. $w = \tau/n$. 

\subsection*{Remark \ref{prop:link}}
 
By Assumption A of Proposition \ref{prop1}, $\mathbb{E}_\mathbb{Q}[S_L(F_i, Y)] = \mathbb{E}_\mathbb{Q}[e_i^2] = c$ for all $i = 1,\ldots,n$. Under squared error loss, $d_{ij} = (\mu_i -\mu_j)^2 = (e_j-e_i)^2$. Hence $\mathbb{E}_\mathbb{Q}[d_{ij}] = \mathbb{E}_\mathbb{Q}[(e_j-e_i)^2] = \mathbb{V}_\mathbb{Q}(e_i-e_j) = 2~c-2~\mathbb{C}\text{ov}_\mathbb{Q}(e_i, e_j),$ where the penultimate equality uses the assumption that $\mathbb{E}_\mathbb{Q}(e_i) = \mathbb{E}_\mathbb{Q}(e_j) = 0$. The assumption of pairwise identical expected divergences is hence equivalent to pairwise identical covariances (and correlations) of forecast errors $e_i, e_j$. 

\subsection*{Proposition \ref{prop:dynamic}}

By definition of $w_{\mathcal{A}}^*$ and the fact that $\mathcal{B} \subseteq \mathcal{A},$ we have
\begin{eqnarray*}
\sum_{i=1}^n w_{\mathcal{A},i}^* \mathbb{E}_{\mathbb{Q}|\mathcal{A}}(S_i|\mathcal{A}) - \sum_{i=1}^{n-1} \sum_{j>i} w_{\mathcal{A},i}^*w_{\mathcal{A},j}^*\mathbb{E}_{\mathbb{Q}|\mathcal{A}}(d_{ij}|\mathcal{A}) &\le& \\
\sum_{i=1}^n w_{\mathcal{B},i}^* \mathbb{E}_{\mathbb{Q}|\mathcal{A}}(S_i|\mathcal{A}) -\sum_{i=1}^{n-1} \sum_{j>i} w_{\mathcal{B},i}^*w_{\mathcal{B},j}^*\mathbb{E}_{\mathbb{Q}|\mathcal{A}}(d_{ij}|\mathcal{A}),&& (\bigstar)
\end{eqnarray*}
$\mathbb{Q}$-almost surely. Using the law of iterated expectations, we obtain 
\begin{eqnarray*}
\mathbb{E}_\mathbb{Q}\left(\sum_{i=1}^n w_{\mathcal{A}, i}^* S_i - \sum_{i=1}^{n-1} \sum_{j>i} w_{\mathcal{A}, i}^*w_{\mathcal{A}, j}^*d_{ij}\right) &=& \mathbb{E}_\mathbb{Q}\left(
\mathbb{E}_{\mathbb{Q}|\mathcal{A}}\left(\sum_{i=1}^n w_{\mathcal{A}, i}^* S_i - \sum_{i=1}^{n-1} \sum_{j>i} w_{\mathcal{A}, i}^*w_{\mathcal{A}, j}^*d_{ij}{\vert} \mathcal{A}\right)\right) \\
&\le &
\mathbb{E}_\mathbb{Q}\left(
\mathbb{E}_{\mathbb{Q}|\mathcal{A}}\left(\sum_{i=1}^n w_{\mathcal{B}, i}^* S_i - \sum_{i=1}^{n-1} \sum_{j>i} w_{\mathcal{B}, i}^*w_{\mathcal{B}, j}^*d_{ij}{\vert} \mathcal{A}\right)\right) \\
&=& \mathbb{E}_\mathbb{Q}\left(\sum_{i=1}^n w_{\mathcal{B}, i}^* S_i - \sum_{i=1}^{n-1} \sum_{j>i} w_{\mathcal{B}, i}^*w_{\mathcal{B}, j}^*d_{ij}\right),
\end{eqnarray*}
where the inequality follows from $(\bigstar)$, and the second equality holds since $\mathcal{B} \subseteq \mathcal{A}$.

\subsection*{Proposition \ref{prop:under}}

From Proposition \ref{prop0} and Assumption \ref{ass:realism}, the expected entropy of the linear pool is given by 
\begin{equation*}
\mathbb{E}_\mathbb{Q}[\text{ent}(F_w, Y)] = \mathbb{E}_\mathbb{Q}[D] + \sum_{i=1}^n w_i~\mathbb{E}_\mathbb{Q}[S(F_i, Y)].\label{eq:first}
\end{equation*}
From Proposition \ref{prop1}, the expected score of the linear pool equals
\begin{equation*}
\mathbb{E}_\mathbb{Q}[S(F_w, Y)] = \sum_{i=1}^n w_i~\mathbb{E}_\mathbb{Q}[S(F_i, Y)] - \mathbb{E}_\mathbb{Q}[D].\label{eq:second}
\end{equation*}
Hence $\mathbb{E}_\mathbb{Q}[\text{ent}(F_w, Y)] - \mathbb{E}_\mathbb{Q}[S(F_w, Y)] = 2~\mathbb{E}_\mathbb{Q}[D] \ge 0$, as stated.

\subsection*{Proposition \ref{lpvsqp}}

Let $F$ be a generic CDF that is strictly increasing on its support $\Omega_F$. Its CRPS entropy is 
\begin{eqnarray}
\text{ent}(F) &=& \int_{-\infty}^\infty F(z) (1-F(z))~dz \nonumber\\
&=& \int_0^1 u(1-u)~dF^{-1}(u), \label{qent}
\end{eqnarray}
using the substitution $u = F(z)$. Consider a given set of forecasts $F_1, \ldots, F_n$ and a fixed weight vector $w$. Since $F^{-1,q}_w = \sum_{i=1}^n w_i F^{-1}_i(u)$, Equation (\ref{qent}) implies that $\text{ent}(F_w^q) = \sum_{i=1}^n w_i~\text{ent}(F_i)$. For the linear pool, Proposition \ref{prop0} implies that $\text{ent}(F_w) \le \sum_{i=1}^n w_i~\text{ent}(F_i)$. Taking expectations establishes the second statement of Proposition \ref{lpvsqp} (smaller expected entropy of quantile combination). To show the first statement (underconfidence of quantile combination), note that
\begin{eqnarray*}
\mathbb{E}_\mathbb{Q}\left(S(F_w^q, Y)\right) &<& \sum_{i=1}^n w_i~ \mathbb{E}_\mathbb{Q}\left(S(F_i, Y)\right) \\
&=& \sum_{i=1}^n w_i~ \mathbb{E}_\mathbb{Q}\left(\text{ent}(F_i)\right) \\
&=& \mathbb{E}_\mathbb{Q}\left(\text{ent}(F_w^q)\right).
\end{eqnarray*}
The strict inequality in the first line follows from \citet[Proposition 1]{LichtendahlEtAl2013}, which establishes a weak inequality, as well as Assumption C of Proposition \ref{lpvsqp}, which rules out the case of equality \citep[see][Proposition 3]{GrushkaEtAl2017}. The second line follows since the $n$ forecasts are realistic, and the last line follows from Equation (\ref{qent}). 

\subsection*{Proposition \ref{prop:qpac}}

By a conditioning argument similar to \citet[Proposition 2]{Krueger2024}, the expected score of an auto-calibrated forecast distribution $F$ is 
\begin{eqnarray*}
\mathbb{E}_\mathbb{Q}[S(F,Y)] &=& \mathbb{E}_\mathbb{Q} \int_{-\infty}^\infty S(F,Y)~d\mathcal{L}(Y|F)\\
&=& \mathbb{E}_Q \int_{-\infty}^\infty S(F,Y)~dF\\
&=& \mathbb{E}_Q[\text{ent}(F)],
\end{eqnarray*}
that is, an auto-calibrated forecast distribution is realistic in the sense of Definition \ref{ass:realism}. Under the stated conditions, quantile-based combination fails to be realistic for the CRPS (see Proposition \ref{lpvsqp}), and thus fails to be auto-calibrated.

\section{Computational Details}

\subsection{CRPS Divergence and Entropy for Gaussian Distributions}\label{crps_gaussian}

Suppose $F$ is a Gaussian distribution with mean $\mu$ and variance $\sigma^2$. Then the CRPS entropy of $F$ is given by
$$\text{ent}(F) = \text{ent}_\mathcal{N}(\mu, \sigma^2) = 0.5~\mathbb{E}|X-\widetilde{X}| = \frac{\sigma}{\sqrt{\pi}};$$
see e.g. \cite{HeldEtAl2010}. Based on similar calculations, the divergence between two Gaussian distributions $F_a, F_b$ is 
\begin{eqnarray*}
d_\mathcal{N}(\mu_a, \sigma^2_a, \mu_b, \sigma^2_b) &=& \mathbb{E}_{F_a, F_b} |X-\widetilde{X}| - 0.5~ \mathbb{E}_{F_a} |X-\widetilde{X}| - 0.5~ \mathbb{E}_{F_b} |X-\widetilde{X}| \\
&=& \sqrt{\sigma_a^2 + \sigma_b^2} \cdot \sqrt{\frac{2}{\pi}} \cdot \exp\left(-\frac{1}{2}~\frac{(\mu_a-\mu_b)^2}{(\sigma_a^2 + \sigma_b^2)}\right) \\
&&+ (\mu_a-\mu_b)\cdot\left(1-2~\Phi\left(-\frac{(\mu_a-\mu_b)}{\sqrt{\sigma_a^2 + \sigma_b^2}}\right)\right)\\
&&-  \frac{\sigma_a}{\sqrt{\pi}} -  \frac{\sigma_b}{\sqrt{\pi}},
\end{eqnarray*}
where $\Phi$ is the cdf of the standard normal distribution. 

\subsection{Empirical Weight Optimization} \label{emp_weight}

Let $\mathbf{S}$ be an $n \times 1$ vector with $i$th element $\mathbf{S}_i = \frac{1}{T} \sum_{t=1}^T S_{i,t},$ where $S_{i,t}$ denotes the score of method $i \in \{1,2,\ldots,n\}$ at date $t \in \{1,2,\ldots, T\}$. Let furthermore $\mathbf{D}$ denote an $n \times n$ matrix with elements $\mathbf{D}_{ij} = \frac{1}{T} \sum_{t=1}^T d_{ij,t},$ where $d_{ij,t}$ is the divergence between methods $i$ and $j$ at date $t$. For an empirical sample of size $T$, the optimal weights are given by
\begin{eqnarray}
\hat w &=& \underset{w \in \text{PS}^n}{\text{arg min}} \left(w'\mathbf{S} - \frac{1}{2} w'\mathbf{D}w\right), \label{v1}
\end{eqnarray}
which is the empirical analogue of Equation (\ref{constweights}). The optimization problem at (\ref{v1}) is not convex quadratic because the matrix $-\mathbf{D}$ is not positive definite. This may be a disadvantage in practice, where some optimization routines require a convex quadratic form. Fortunately, the following proposition shows that (\ref{v1}) is equivalent to a convex quadratic problem that can be used in practice. 

\begin{prop}
Let $\tau$ denote an $n \times 1$ vector of ones, and let $\mathbf{a} = \mathbf{S}-\mathbf{c}$ and $\mathbf{B} = \mathbf{c}\tau' + \tau \mathbf{c}' - \mathbf{D},$ where $\mathbf{c}$ is an $n \times 1$ vector with $i$th element $\mathbf{c}_i = \frac{1}{T} \sum_{t=1}^T S(F_{i,t}, 0),$ i.e. the average score of method $i$ for hypothetical zero outcomes $y_1 = y_2 = \ldots = y_T = 0$. Then the objective function
\begin{eqnarray}
\left(w'\mathbf{a} + \frac{1}{2} w'\mathbf{B}w\right), \label{v2}
\end{eqnarray}
(i) is equivalent to the objective function at (\ref{v1}) and (ii) takes the form of a convex quadratic program, i.e. the matrix $\mathbf{B}$ is positive definite. 
\end{prop}
\begin{proof}
(i) follows by substituting the definitions of $\mathbf{a}$ and $\mathbf{B}$ into (\ref{v2}) and using the fact that $\tau'w = 1$. To show (ii), denote the $i,j$ element of $\mathbf{B}$ by 
$$\mathbf{B}_{ij} = \frac{1}{T} \sum_{t=1}^T b_{ij,t},$$
where 
\begin{eqnarray*}
b_{ij,t} &=& \mathbb{E} L(X_{i,t}, 0) - 0.5~\mathbb{E} L(X_{i,t}, \widetilde X_{i,t}) + \mathbb{E} L(X_{j,t}, 0) - 0.5~\mathbb{E} L(X_{j,t}, \widetilde X_{j,t}) - d_{ij,t}\\
&=& \mathbb{E} L(X_{i,t}, 0) + \mathbb{E} L(X_{j,t}, 0) - \mathbb{E} L(X_{i,t}, \widetilde X_{j,t}),
\end{eqnarray*}
where $X_{i,t}$ is distributed according to $F_{i,t}$. By comparison with \citet[Proposition 2]{AllenEtAl2024}, this matrix $\mathbf{B}$ corresponds to the weight optimization problem for the kernel function $k(a, b) = L(a, 0) + L(b, 0) - L(a, b).$ Since $L$ is negative definite and $L(0, 0) = 0$, the kernel function $k$ is positive definite \citep[Lemma 12]{SejdinovicEtAl2013}, so that the objective function at (\ref{v2}) is convex quadratic. 

\end{proof}

\section{Additional Figures for Section \ref{sec:empirical}} \label{sec:addfig}

\begin{figure}[!htbp]
	\includegraphics[width=\textwidth]{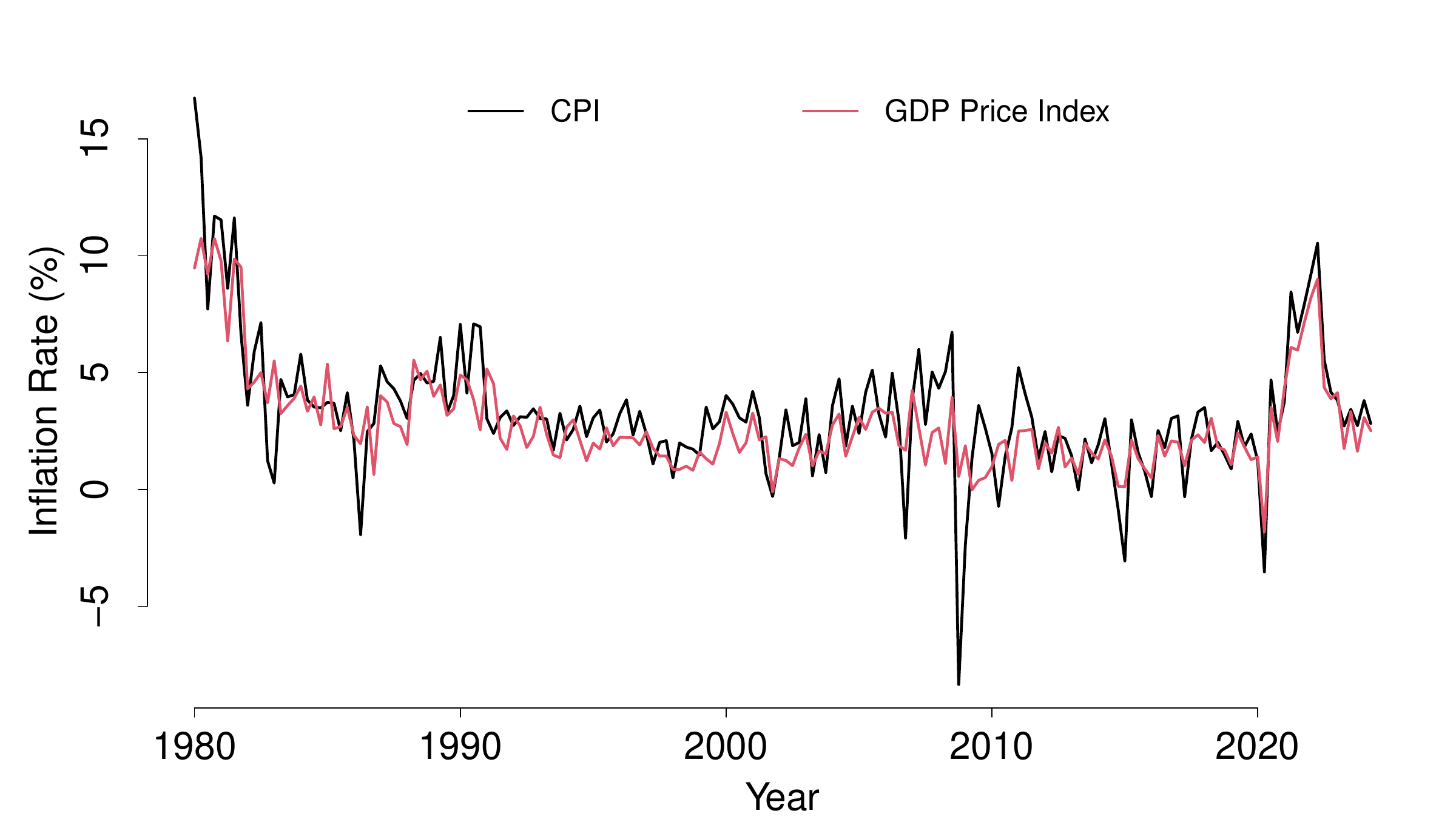}
	\caption{Time series of the two inflation measures considered in Section \ref{sec:illustration}. Inflation rates are computed as annualized quarterly growth rates of the underlying index. For each quarter, we use the second vintage available in the Philadelphia Fed's real-time database. \label{fig:inflation_actuals}}
\end{figure}

\begin{figure}
	\includegraphics[width=\textwidth]{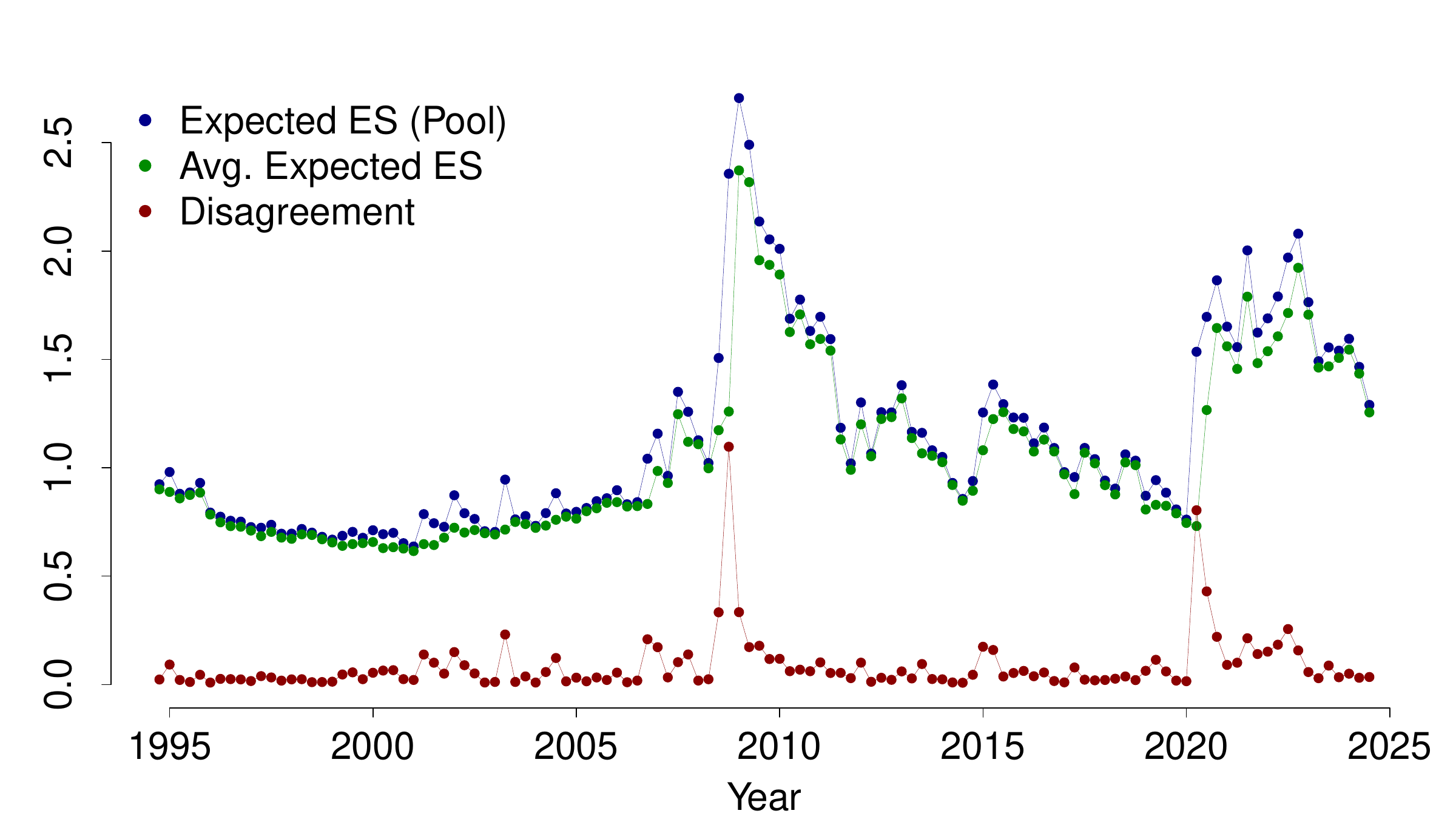}
	\caption{Illustration of the decomposition in Equation (\ref{entropy_pool}). The figure shows the equally weighted linear pool's expected Energy Score and its components for current-quarter forecast distributions of inflation. \label{fig:es_spf_bvar_h1}}
\end{figure}

\begin{figure}
	\includegraphics[width=\textwidth]{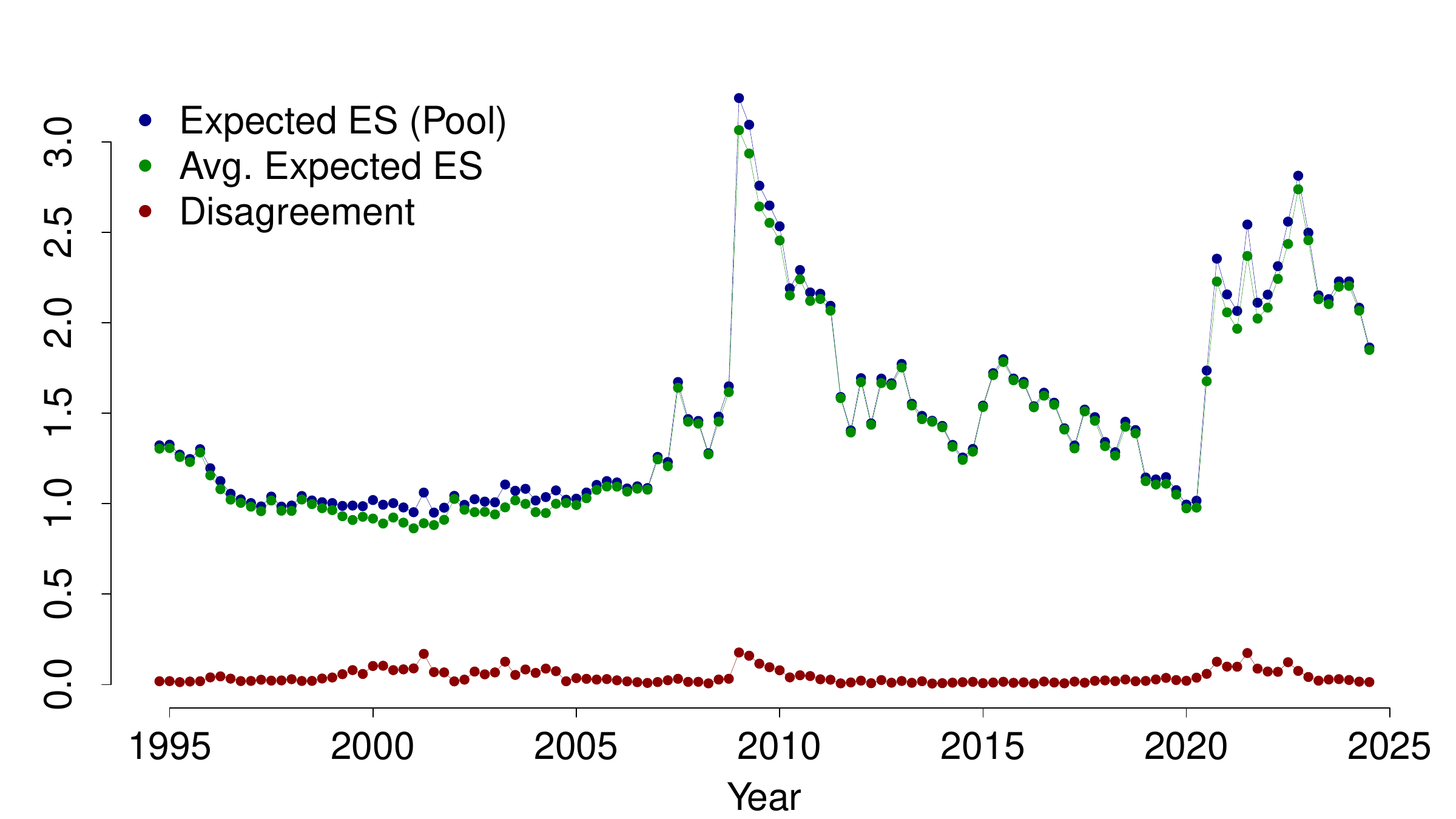}
	\caption{Like Figure \ref{fig:es_spf_bvar_h1}, but for horizon $h = 5$. \label{fig:es_spf_bvar_h5}}
\end{figure}

\begin{figure}
\includegraphics[width=\textwidth]{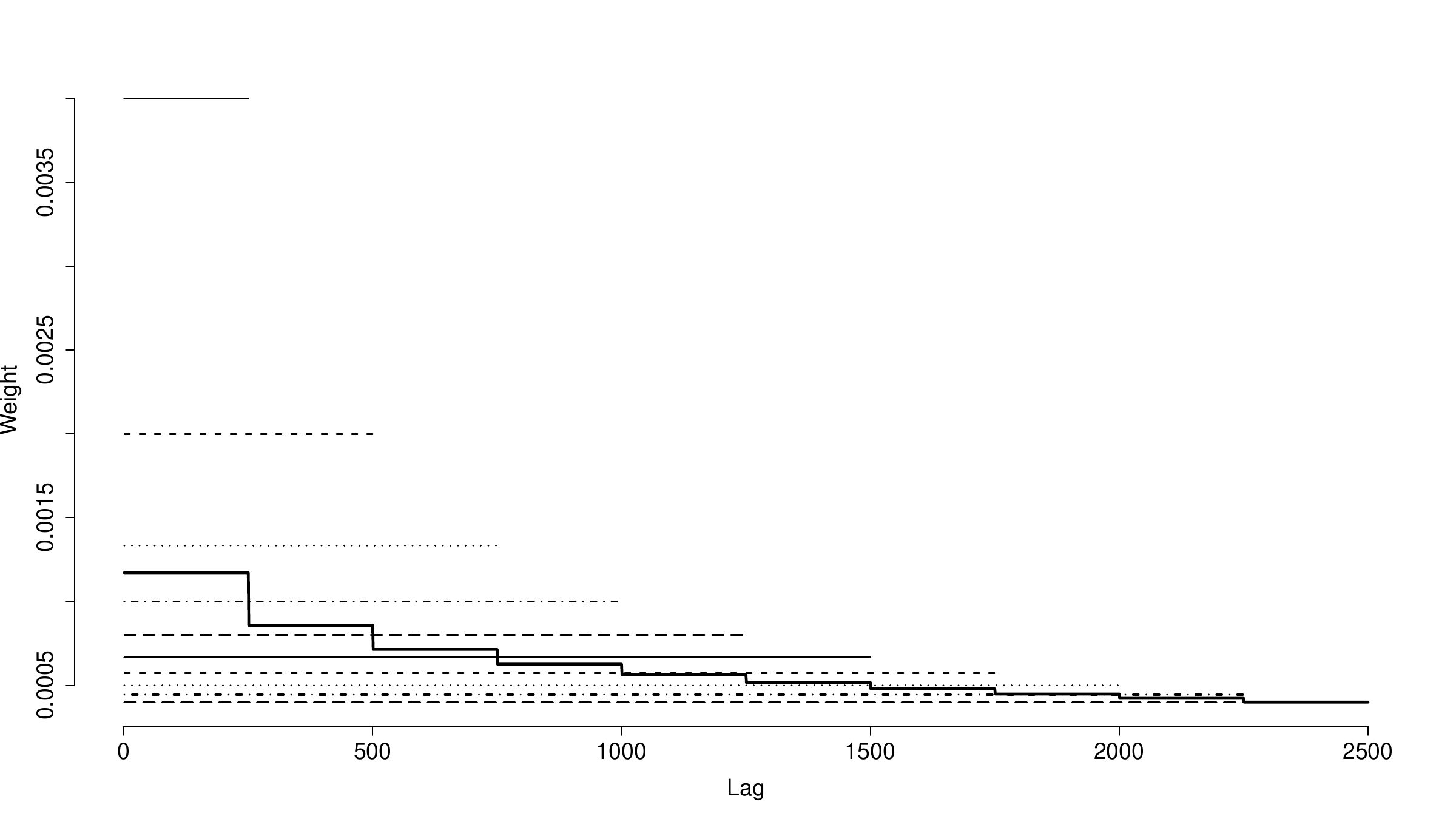}
\caption{Plot of weight functions against lag order. Lag order $j$ denotes observation $y_{t-1-j}$, where $t-1$ is the date at which the forecast is formed. Ten horizontal lines correspond to individual choices. Step function corresponds to equally weighted mean of the ten weight functions. \label{fig:lags}}
\end{figure}

\end{appendix}

\end{document}